\documentclass{amsart}

\usepackage[utf8]{inputenc}
\usepackage{amssymb}
\usepackage{graphicx}
\usepackage{xcolor}
\usepackage{hyperref}

\usepackage{tikz}
\usetikzlibrary{arrows, automata, shapes, positioning}
\usetikzlibrary{arrows.meta}
\usetikzlibrary{decorations.pathmorphing,patterns,decorations.pathreplacing}
\usetikzlibrary{calc}

\DeclareMathOperator{\rep}{rep}
\DeclareMathOperator{\val}{val}
\DeclareMathOperator{\dom}{dom}

\theoremstyle{plain}
\newtheorem{theorem}{Theorem}
\newtheorem{lemma}[theorem]{Lemma}
\newtheorem{corollary}[theorem]{Corollary}
\newtheorem{prop}[theorem]{Proposition}

\theoremstyle{definition}
\newtheorem{definition}[theorem]{Definition}
\newtheorem{remark}[theorem]{Remark}

\newtheorem{example}[theorem]{Example}

\numberwithin{equation}{section}

\title[Revisiting regular sequences]{Revisiting regular sequences in light of rational base numeration systems}
\author{Michel Rigo and Manon Stipulanti}
\address{
Department of Mathematics\\
University of Li\`ege\\
All\'ee de la D\'ecouverte 12\\
4000 Li\`ege, Belgium\\
{\tt \{m.rigo,m.stipulanti\}@uliege.be}}

\begin{document}

\maketitle

\begin{abstract}
Regular sequences generalize the extensively studied automatic sequences. 
Let $S$ be an abstract numeration system.
When the numeration language $L$ is prefix-closed and regular, a sequence is said to be $S$-regular if the module generated by its $S$-kernel is finitely generated.

In this paper, we give a new characterization of such sequences in terms of the underlying numeration tree $T(L)$ whose nodes are words of $L$.
We may decorate these nodes by the sequence of interest following a breadth-first enumeration.
For a prefix-closed regular language $L$, we prove that a sequence is $S$-regular if and only if the tree $T(L)$ decorated by the sequence is linear, i.e., the decoration of a node depends linearly on the decorations of a fixed number of ancestors.

Next, we introduce and study regular sequences in a rational base numeration system, whose numeration language is known to be highly non-regular.
We motivate and comment our definition that a sequence is $\frac{p}{q}$-regular if the underlying numeration tree decorated by the sequence is linear.
We give the first few properties of such sequences, we provide a few examples of them, and we propose a method for guessing $\frac{p}{q}$-regularity.
Then we discuss the relationship between $\frac{p}{q}$-automatic sequences and $\frac{p}{q}$-regular sequences. 
We finally present a graph directed linear representation of a $\frac{p}{q}$-regular sequence.  Our study permits us to highlight the places where the regularity of the numeration language plays a predominant role.
\end{abstract}

\bigskip
\hrule
\bigskip



\bigskip
\hrule
\bigskip

\section{Introduction}\label{sec: intro}

The notion of $k$-regular sequences \cite{Ring} is a natural generalization of $k$-automatic sequences \cite{Rowland} whenever the set of taken values is infinite. It means that the $k$-kernel of the sequence is included in a finitely generated module (or simply a finite-dimensional $\mathbb{Q}$-vector space when dealing with integer sequences). Otherwise stated, a sequence $(x_n)_{n\ge 0}$ is {\em $k$-regular} if and only if there exists an integer $L\ge 0$ such that any subsequence $(x_{k^\ell n+r})_{n\ge 0}$ of its $k$-kernel can be expressed, for all $\ell>L$ and $0\le r<k^\ell$, as a linear combination of a finite number of subsequences of the form  $(x_{k^m n+j})_{n\ge 0}$ for $0\le m\le L$ and $0 \le j< k^m$. This concept has proven its usefulness in many fields: number theory, combinatorics on words, enumeration or discrete mathematics and theoretical computer science or algorithm analysis \cite{Ring,AS2,BR2,CRS}. The notion of kernel and regular sequences have been generalized to non-standard numeration systems with a regular numeration language \cite{MR,CCS}.

In a recent article \cite{RS}, we introduced automatic sequences built on a rational base numeration system. In this framework, contrarily to the classical situation where the numeration language is regular, the sequences are built from representations belonging to a highly non-regular language. In the current paper, our aim is to extend the notion of regular sequences to this framework. Our study permits us to highlight the places where the regularity of the numeration language plays a predominant role. Because of its universality, the case of integer bases does not always reveal the importance of the underlying language.

The simplest example of what should be a $\frac32$-regular sequence is probably the sum-of-digits sequence in base $\frac{3}{2}$ taking values in $\mathbb{N}$:
$$\mathbf{s}=(s_n)_{n\ge 0}=0, 2, 3, 3, 5, 4, 5, 7, 5, 5, 7, 8, 5, 7, 6, 7, 9, 9,\ldots.$$
Indeed, the representations of the first few integers in base~$\frac{3}{2}$ are $\varepsilon$, $2$, $21$, $210$, $212$, $2101$, $2120$, $2122$, $21011$, $21200$, \ldots. If, like Allouche and Shallit \cite{Ring}, we consider the kernel formed of all subsequences of $\mathbf{s}$ whose selected indices have a representation in base~$\frac{3}{2}$ sharing a fixed suffix, since the prefix-tree of the base-$\frac{3}{2}$ numeration language is built on a periodic signature \cite{Marsault--Sakarovitch-2016}, we obtain the ``classical'' $3$-kernel of the sequence: every third $\frac32$-representation in radix order has the same last digit, every ninth $\frac32$-representation has the same suffix of length~$2$, and so on and so forth. Since $\mathbf{s}$ is unbounded, the set of subsequences of $\mathbf{s}$ of the form $(s_{3^jn+r})_{n\ge 0}$ is infinite, and moreover it has no clear combinatorial properties (as mentioned above, one usually looks for linear relations between elements of the kernel). A similar negative observation is made when computing elements of its $2$-kernel or using the generalized definition of kernel from \cite{CCS}. (See Section~\ref{sec: p/q-regularite} for more details.) Hence the usual definitions of regularity based on kernel sets do not provide us with a suitable setting. We note that the classical framework is thus not applicable.

It is therefore necessary to find an alternative to the definition involving the finitely generated property of the module generated by the kernel. Our motivation is to introduce a new definition of a regular sequence in a rational base numeration system, which also extends the classical framework of integer base numeration systems or numeration systems built on a prefix-closed regular language. Our approach is based on the structure of the prefix-tree of the numeration language. Indeed, since rational base numeration languages are prefix-closed, they are conveniently represented by trees whose nodes are in one-to-one correspondence with the words of the language. Note that the prefix-closed property is also assumed in \cite{CCS}. 
  The nodes of such a tree are then decorated by the sequence $\mathbf{x}=(x_n)_{n\ge 0}$ of interest by breadth-first enumeration (or serialization): the $n$th node representing $n$ has decoration $x_n$. We introduce the concept of a {\em linear decorated tree} where the decoration of a node depends linearly on the decorations of a fixed number of ancestors.

  The reader accustomed to $k$-regular sequences should not be surprised by our description given in terms of prefix-trees in which are found linear relationships. Indeed, this is somehow a reformulation of the fact that linear relations occur among elements of the $k$-kernel. However, compared with the definition based on the kernel, our approach allows us to propose a definition extending regularity to numeration systems with non-regular (prefix-closed) languages. As an example, the sum-of-digits sequence $\mathbf{s}$ gives rise to a decorated linear tree but its kernel does not belong to a finitely generated module.   Moreover, we are able to propose a kind of linear representation that permits us to compute the decoration associated with $x_n$ with a number of matrix multiplications equal to the length of the representation of $n$ in the considered rational base numeration system.

Our study also highlights the limitations of a numeration system built on a non-regular language. With our definition of regularity extending the classical framework, we show that any regular sequence taking a finite number of values is automatic. However, the converse is false. One can easily construct a sequence produced by an automaton whose outputs are all distinct, which is not regular.

Berstel and Reutenauer \cite{BR2} present the $k$-regularity of a sequence $(x_n)_{n\ge 0}$ through the fact that the formal series $\sum_{n\ge 0} x_n \rep_k(n)$ is rational (where $\rep_k(n)$ denotes the base-$k$ representation of $n$). In a more general setting, one considers the series $\sum_{n\ge 0} x_n \rep_S(n)$ where $\rep_S(n)$ is the representation of $n$ is a convenient abstract numeration system $S$ and the support of these series is thus included in the numeration language.  With the above example of the sum-of-digits sequence $\mathbf{s}$, such a series (each term is made of a coefficient followed by a word of the language) starts with
$$\sum_{n\ge 0} s_n \rep_{\frac{3}{2}}(n)=2\cdot 2+ 3\cdot 21+ 3\cdot 210+5\cdot 212+4\cdot 2101+5\cdot 2120+\cdots.$$
So in the case of $\frac32$-representations, we have a series with a non-regular support. In particular, this series cannot be $\mathbb{N}$-rational~\cite[Chapter~3]{BR2}.
In this paper, we are extensively considering trees and subtrees but we are not considering series as functions defined over the set of finite trees (i.e., with bounded height) and taking values in a semiring. Indeed, we point out that such a variant of rational series on trees has been considered in \cite{BR} and is an independent matter. The map associating with a tree its height or the evaluation of a parse tree of an arithmetic expression are examples of this kind. The authors studied the corresponding notion of rational series.

The paper is organized as follows.
In Section~\ref{sec: arbres}, we recall the necessary background on trees. In particular, we distinguish decorated and undecorated trees.
Section~\ref{sec: arbres decores} is dedicated to our main concept, namely linear trees.
Roughly a tree is $h$-linear for some $h\ge 0$ if, for each type of undecorated subtrees of height $h$, the decoration of each node on level $h$ can be obtained as a linear combination of decorations of nodes on levels $0, 1,\ldots, h-1$. 
We provide examples of such trees in three different numeration systems (integer base, Fibonacci or Zeckendorf system and rational base $\frac{3}{2}$).
In Section~\ref{sec: regularite et arbres}, with Theorems~\ref{the:kautom} and \ref{the:Sreg}, given an abstract numeration system $S$ based on a prefix-closed regular language, we give an alternative characterization of regularity by showing that a sequence is $S$-regular if and only if the decorated tree associated with the numeration language is linear.
Then, in Section~\ref{sec: p/q-regularite}, we explain why such an equivalence does not hold for a rational base whose numeration language is not regular. We thus argue that a meaningful definition for regularity in rational bases should be related to the linearity of the decorated numeration tree. This provides a new framework to extend the notion further beyond rational bases.
We give the first few properties of regular sequences in rational bases. As Allouche and Shallit in \cite{AS2}, we provide a method for guessing $\frac{p}{q}$-regularity. Thanks to this technique, we provide a few examples of $\frac{p}{q}$-regular sequences such as $(n^d)_{n\ge 0}$ for any integer $d\ge 0$.
Then we compare the set of $\frac{p}{q}$-automatic sequences with the set of $\frac{p}{q}$-regular sequences taking finitely many values. 
We make use of our results to show that the cumulative version (defined later in the paper) of a $\frac{p}{q}$-regular sequence is also $\frac{p}{q}$-regular.
Finally in Section~\ref{sec: matrices}, we provide a graph directed linear representation of a regular sequence in rational bases. This again highlights the main differences between a regular and a non-regular numeration languages.

\section{Trees}\label{sec: arbres}

Let $L$ be a prefix-closed language over a totally ordered alphabet $(A,<)$ and let $\mathbf{x}=x_0x_1\cdots$ be an infinite sequence over a commutative semiring $\mathbb{K}$, i.e., $\mathbf{x}\in\mathbb{K}^\mathbb{N}$. With $L$ is associated a tree whose (linearly ordered) nodes are ``colored'' by the terms of $\mathbf{x}$. The tree representation of the free monoid $A^*$ is classical; we consider a subtree of this infinite tree.

\begin{definition}
  With every prefix-closed language $L\subseteq A^*$ is associated a labeled tree $T(L)$ whose nodes are the words of $L$. The empty word $\varepsilon$ is the root. Edges of the tree are labeled by letters in $A$. The alphabet is assumed to be totally ordered, so the edges are also ordered. If $u$ and $ua$ are two nodes with $a\in A$, there is an edge of label $a$ between them. 
\end{definition}

Enumerating the words of $L$ by radix order (so considering an abstract numeration system $S=(L,A,<)$ in the sense of Lecomte and Rigo \cite{LR}) corresponds to the breadth-first traversal (or serialization \cite{Marsault--Sakarovitch-2}) of the ordered tree $T(L)$. For all $n\ge 0$, the $n$th word $w_n$ in $L$ corresponds to the $n$th node of $T(L)$. We let $\rep_S(n)$ denote the $(n+1)$st word in $L$ for all $n\ge 0$, and $\val_S:L\to\mathbb{N}$ is the inverse function mapping any word of $L$ to its position in the radix ordered language $L$.

\begin{definition}
A {\em decoration} of a tree $T=(V,E)$ is a map from the set $V$ of vertices to some set $B$ associating a value in $B$ with each node of the tree.
\end{definition}

To avoid any confusion or misunderstanding, we say that edges have {\em labels} and nodes have {\em decorations}. In a tree, nodes of {\em level~$\ell\ge 0$} are those at distance $\ell$ from the root.

 From now on, we consider the labeled tree $T(L)$ decorated by $\mathbf{x}$.  It is denoted by $T_\mathbf{x}(L)$. Otherwise stated, for the abstract numeration system built on $L$, the node corresponding to a word $w\in L$ in $T_\mathbf{x}(L)$ has decoration $x_{\val(w)}$. By abuse of notation, since integers are in one-to-one correspondence with words in $L$, we also write $x_w$. In Figure~\ref{fig:dec_tree}, we have depicted the first few levels of the tree of the numeration language associated with base-$2$, Fibonacci and base-$\frac32$ numeration systems. The decorations are given by the sequence $(n)_{n\ge 0}$. 

\begin{figure}[h!t]
  \centering
\begin{minipage}{.4\linewidth}
  \tikzset{
  s_gra/.style = {circle,fill=gray!30, inner sep=1pt, minimum size=15pt},
}
\tikzstyle{level 1}=[sibling distance=50mm]
\tikzstyle{level 2}=[sibling distance=33mm]
\tikzstyle{level 3}=[sibling distance=13mm]
\tikzstyle{level 4}=[sibling distance=7mm]
\begin{tikzpicture}[->,>=stealth',level distance = 1cm]
  \node [s_gra] {$0$} 
  child { node [s_gra] {$1$} 
    child { node [s_gra] {$2$} 
      child { node [s_gra] {$4$} 
        child { node [s_gra] {$8$} 
          edge from parent node[left] {$0$} }
        child { node [s_gra] {$9$} 
          edge from parent node[right] {$1$} }
        edge from parent node[left] {$0$} }
      child { node [s_gra] {$5$} 
        child { node [s_gra] {$10$} 
          edge from parent node[left] {$0$} }
        child { node [s_gra] {$11$} 
          edge from parent node[right] {$1$} }
        edge from parent node[right] {$1$} }
      edge from parent node[above] {$0$} } 
    child { node [s_gra] {$3$} 
      child { node [s_gra] {$6$} 
        child { node [s_gra] {$12$} 
          edge from parent node[left] {$0$} }
        child { node [s_gra] {$13$} 
          edge from parent node[right] {$1$} }
        edge from parent node[left] {$0$} }
      child { node [s_gra] {$7$} 
        child { node [s_gra] {$14$} 
          edge from parent node[left] {$0$} }
        child { node [s_gra] {$15$} 
          edge from parent node[right] {$1$} }
        edge from parent node[right] {$1$} }
      edge from parent node[above] {$1$} } 
    edge from parent node[left] {$1$}
  }
;
\end{tikzpicture}
\end{minipage}\hskip1.5cm  
\begin{minipage}{.2\linewidth}
\tikzset{
  s_gra/.style = {circle,fill=gray!30, inner sep=1pt, minimum size=15pt},
}
\begin{tikzpicture}[->,>=stealth',level/.style={sibling distance = 4cm/#1},level distance = 1cm]
  \node [s_gra] {$0$} 
  child { node [s_gra] {$1$} 
    child { node [s_gra] {$2$} 
      child {node [s_gra] {$3$} 
        child {node [s_gra] {$5$} edge from parent node[left] {$0$}}
        child {node [s_gra] {$6$} edge from parent node[right] {$1$}} 
            edge from parent node[left] {$0$}}
    child {node [s_gra] {$4$} 
      child {node [s_gra] {$7$} edge from parent node[right] {$0$}} 
  edge from parent node[right] {$1$}}
edge from parent node[left] {$0$}}
edge from parent node[left] {$1$}}  
;
\end{tikzpicture}
\end{minipage}\hskip1cm  
\begin{minipage}{.2\linewidth}
\tikzset{
  s_gra/.style = {circle,fill=gray!30, inner sep=1pt, minimum size=15pt},
}
\begin{tikzpicture}[->,>=stealth',level/.style={sibling distance = 4cm/#1},level distance = 1cm]
  \node [s_gra] {$0$} 
  child { node [s_gra] {$1$} 
    child { node [s_gra] {$2$} 
      child {node [s_gra] {$3$} 
      child {node [s_gra] {$5$} edge from parent node[left] {$1$}} 
            edge from parent node[left] {$0$}}
    child {node [s_gra] {$4$} 
      child {node [s_gra] {$6$} edge from parent node[left] {$0$}} 
      child {node [s_gra] {$7$} edge from parent node[right] {$2$}} 
  edge from parent node[right] {$2$}}
edge from parent node[left] {$1$}}
edge from parent node[left] {$2$}}  
;
\end{tikzpicture}
\end{minipage}
\caption{Prefixes of height $4$ of three trees $T(L)$ where $L$ is respectively the base-$2$, Fibonacci and base-$\frac{3}{2}$ numeration language.}
  \label{fig:dec_tree}
\end{figure}
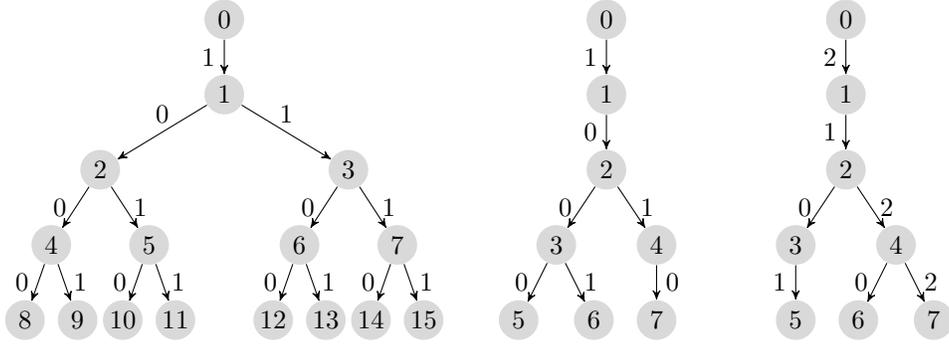

\begin{definition}
The {\em domain} $\dom(T)$ of a labeled tree $T$ is the set of labels of paths from the root to its nodes. In particular, $\dom(T(L))=L$. The {\em truncation} of a tree at height~$h$ is the restriction of the tree to the domain $\dom(T)\cap A^{\le h}$.
\end{definition}

\begin{definition}
  Let $w\in L$. We let $T[w]$ denote the (infinite) subtree of $T(L)$ having $w$ as root. Its domain is $w^{-1}L=\{u \mid wu\in L\}$. We say that $T[w]$ is a {\em suffix} of~$T$. The {\em factor of height $h$ rooted at $w$} is the truncation of $T[w]$ at height~$h$. It is denoted by $T[w,h]$. The {\em prefix of height $h$} of $T$ is the factor $T[\varepsilon,h]$. We distinguish the case without or with decorations.
  \begin{itemize}
  \item Two factors $T[w,h]$ and $T[w',h]$ of the same height in the (undecorated) tree $T(L)$ are {\em equal} if they have the same domain. 
  \item Two factors $T[w,h]$ and $T[w',h]$ of the same height in the decorated tree $T_\mathbf{x}(L)$  are {\em equal} if they have the same domain and the same decorations, i.e., $x_{wu}=x_{w'u}$ for all $u\in \dom(T[w,h])$.
      \end{itemize}
\end{definition}

For instance, in the first tree of Figure~\ref{fig:dec_tree}, the factors $T[10,2]$ and $T[11,2]$ are equal as undecorated trees (they have the same domain $\{0,1\}^{\le 2}$), but are different as decorated trees.

\begin{lemma}\label{lem:auto}
  If $L$ is a regular language, then the non-decreasing function mapping $n$ to the number of (undecorated) factors of height $n$ in $T(L)$ is bounded by the number of states of the trim minimal automaton of $L$.
\end{lemma}

\begin{proof}
  Any two states in the minimal automaton of $L$ are distinguished: there is a word accepted from exactly one of the two states and rejected from the other one.
Moreover, it is a classical result that if two states are distinguishable, then they are distinguished by a word of length at most $N-1$ where $N$ is the number of states of the minimal automaton of $L$. 
  Otherwise stated, if the word $w$ (respectively $w'$) leads to a state $q$ (respectively $q'$), then there is a word $u$ of length $h\le N-1$ such that $u$ belongs to the domain $\dom(T[w,h])=w^{-1} L\cap A^{\le h}$ of $T[w,h]$ or to the domain $\dom(T[w',h])={w'}^{-1} L\cap A^{\le h}$ of $T[w',h]$ but not to both of them.
\end{proof}

\begin{remark}\label{rem:isoTw}
Let $L$ be a regular language having $\mathcal{M}_L$ as minimal automaton. From the above lemma, if the height $h$ is large enough, e.g., at least the number of states of $\mathcal{M}_L$, there is a one-to-one correspondence between the states of $\mathcal{M}_L$ and the pairwise distinct factors of height $h$ in $T(L)$. In particular, $T[w,h]=T[w',h]$ if and only if, from the initial state, the words $w$ and $w'$ lead to the same state of  $\mathcal{M}_L$.
\end{remark}

\section{Decorated $h$-linear trees}\label{sec: arbres decores}

The next definition is central and has to be understood as follows. Roughly speaking, whenever two factors of height $h$ occurring in $T(L)$ are isomorphic as labeled trees, then the decorations of the nodes on the lowest level~$h$ are obtained as linear combinations of the decorations of the nodes on the other levels. For nodes on level $h$ at the same relative positions in the two trees, the linear combinations to get the corresponding decorations are the same.

Let $\mathbb{L}$ be a commutative Noetherian ring, i.e., a ring in which every ideal is finitely generated. We assume that $\mathbb{L}$ is a subsemiring of $\mathbb{K}$. So we will consider $\mathbb{L}$-linear combinations of elements in $\mathbb{K}^\mathbb{N}$. We make such an assumption because in general, a submodule of a finitely generated module is not always finitely generated, but it is the case when the module is over a Noetherian ring. In particular, principal ideal domains (PID for short) and fields, such as $\mathbb{Z}$ and $\mathbb{Q}$ respectively, are Noetherian rings.

\begin{definition}  Let $L$ be any prefix-closed language (not necessarily regular) over a finite alphabet $A$. Fix some integer $h\ge 1$ and let $r(h)\ge 1$ denote the number of pairwise distinct labeled trees of height $h$ occurring in $T(L)$, denoted by $T_1,\ldots,T_{r(h)}$.
  The decorated tree $T_\mathbf{x}(L)$ is {\em $(\mathbb{L},h)$-linear} (or simply {\em $h$-linear}) if, for all $1\le j\le r(h)$ and all words $w\in L$ such $T[w,h]=T_j$, 
  then for all words $u$ belonging to $\dom(T_j) \cap A^h$,
  there exist constants $c_{j,u,v}\in\mathbb{L}$ such that 
$$x_{wu}=\sum_{\substack{v\in \dom(T_j)\\ |v|<h}} c_{j,u,v}\, x_{wv}.$$
In what follows, for all words $w\in L$ such that the height-$h$ subtree $T[w,h]$ of $T(L)$ is equal to $T_j$ for some $1\le j\le r(h)$, we say that the \emph{type} of $T[w,h]$ is $T_j$.
\end{definition}

\begin{lemma}\label{lem:h+1}
  If a decorated tree $T_\mathbf{x}(L)$ is $h$-linear, then it is $(h+1)$-linear.
\end{lemma}

\begin{proof}
  Let $T[w,h+1]$ be a subtree of height $h+1$ in $T_\mathbf{x}(L)$. Assume that the words $wa_1,\ldots,wa_\ell$ are the children of $w$ with $a_j \in A$ for all $1\le j \le \ell$. Hence $T[w,h+1]$ contains $\ell$ disjoint (i.e., with no common node) subtrees $T[wa_1,h],\ldots,T[wa_\ell,h]$ of height $h$. By assumption, for all $1\le j \le \ell$, the decorations on the last level of $T[wa_j,h]$ are linear combinations of the decorations of the nodes in $T[wa_j,h-1]$, which are nodes of $T[w,h]$. Otherwise stated, the decorations on the last level of $T[w,h+1]$ are linear combinations of the decorations of the nodes in $T[w,h]$.
\end{proof}

If a decorated tree $T_\mathbf{x}(L)$  is $h$-linear, it is enough to know its prefix of height $h$ and the linear relations occurring on the lower level in every non-isomorphic factor of height $h$ to recover the full decorated tree. Indeed, from the prefix of height $h-1$ of any factor of height~$h$, one can compute the decorations on the next level.
We now illustrate this fact on three different types of numeration systems.

Let $k\ge 2$ be an integer. We let $A_k$ be the finite alphabet $\{0,1,\ldots,k-1\}$. The usual base-$k$ numeration system is built on the language
\begin{equation}
  \label{eq:Lk}
L_k:=\{\varepsilon\}\cup \{1,\ldots,k-1\}\{0,\ldots,k-1\}^*.  
\end{equation}

\begin{example}[Integer base]
  Assume that the decorations take values in $\mathbb{N}$ and that we consider the tree $T(L_2)$ in the binary numeration system. From the prefix and the linear combinations given in Figure~\ref{fig:ex1}, we can build a decorated infinite $3$-linear tree. We proceed as follows. Start with the given decorated prefix of height~$3$. In the factor $T[1,3]$, the first three levels are already decorated. If we make use of the given relations, we can compute the decorations on level~$4$ as depicted in Figure~\ref{fig:ex1b}. For instance, the fourth node on level 4 has decoration $0+2-1=1$. We can continue with the factor $T[10,3]$ (and respectively $T[11,3]$) where the first three levels are already decorated. We use the same relations and obtain the decorations on level~$5$, and so on and so forth.
  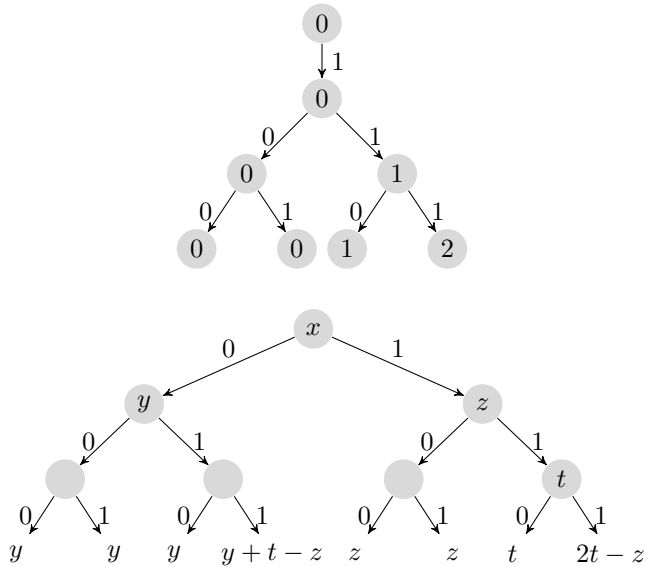
\begin{figure}[h!tb]
    \centering
\tikzset{
  s_gra/.style = {circle,fill=gray!30, inner sep=1pt, minimum size=15pt},
}
\begin{tikzpicture}[->,>=stealth',level/.style={sibling distance = 4cm/#1},level distance = 1cm]
  \node [s_gra] {$0$} 
  child { node [s_gra] {$0$} 
    child { node [s_gra] {$0$}
      child { node [s_gra] {$0$} edge from parent node[left] {$0$}}
      child { node [s_gra] {$0$} edge from parent node[right] {$1$}}
      edge from parent node[left] {$0$}}
    child { node [s_gra] {$1$}
      child { node [s_gra] {$1$} edge from parent node[left] {$0$}}
      child { node [s_gra] {$2$} edge from parent node[right] {$1$}}
      edge from parent node[right] {$1$}}
edge from parent node[right] {$1$}}  
;
\end{tikzpicture}\vskip.5cm
\tikzset{
  s_end/.style = {inner sep=1pt, minimum size=15pt},
  s_gra/.style = {circle,fill=gray!30, inner sep=1pt, minimum size=15pt}
}
\tikzstyle{level 1}=[sibling distance=45mm]
\tikzstyle{level 2}=[sibling distance=21mm]
\tikzstyle{level 3}=[sibling distance=13mm]
\begin{tikzpicture}[->,>=stealth',level/.style={sibling distance = 3cm/#1},level distance = 1cm]
  \node [s_gra] {$x$} 
  child { node [s_gra] {$y$} 
    child { node [s_gra] {}
      child { node [s_end] {$y$} edge from parent node[left] {$0$}}
      child { node [s_end] {$y$} edge from parent node[right] {$1$}}
      edge from parent node[left] {$0$}}
    child { node [s_gra] {}
      child { node [s_end] {$y$} edge from parent node[left] {$0$}}
      child { node [s_end] {$y+t-z$} edge from parent node[right] {$1$}}
      edge from parent node[right] {$1$}}
    edge from parent node[above] {$0$}}
   child { node [s_gra] {$z$} 
     child { node [s_gra] {}
       child { node [s_end] {$z$} edge from parent node[left] {$0$}}
       child { node [s_end] {$z$} edge from parent node[right] {$1$}}
       edge from parent node[left] {$0$}}
     child { node [s_gra] {$t$}
       child { node [s_end] {$t$} edge from parent node[left] {$0$}}
       child { node [s_end] {$2t-z$} edge from parent node[right] {$1$}}
       edge from parent node[right] {$1$}}
edge from parent node[above] {$1$}}  
;
\end{tikzpicture}
    \caption{A first example (base-$2$ case): the prefix of height $3$ and linear relations to extend the tree.}
    \label{fig:ex1}
  \end{figure}
  We obtain the tree in Figure~\ref{fig:ex1b} and the serialization of the decorations is the sequence $0,0,0,1,0,0,1,2,0,0,0,1,1,1,2,3,\ldots$ {\cite[A014081]{OEIS}}, which counts the number of occurrences of $11$ (with possible overlaps) in the binary representation of $n$. To prove that we indeed get this sequence, observe that from a node decorated by $z$ to the one decorated by $t$, an extra $11$ has been read. Thus, $t-z$ should be equal to $1$ so, each time a new factor $11$ occurs, $t-z$ has to be added to the current decoration. This explains the relations on the lowest level in Figure~\ref{fig:ex1}.  
  \begin{figure}[h!tb]
    \centering
\begin{minipage}{.4\linewidth}
  \tikzset{
  s_gra/.style = {circle, draw, inner sep=1pt, minimum size=15pt},
}
\tikzstyle{level 1}=[sibling distance=50mm]
\tikzstyle{level 2}=[sibling distance=33mm]
\tikzstyle{level 3}=[sibling distance=13mm]
\tikzstyle{level 4}=[sibling distance=7mm]
\begin{tikzpicture}[->,>=stealth',level distance = 1cm]
  \node [s_gra] {$0$} 
  child { node [s_gra] {$0$} 
    child { node [s_gra] {$0$} 
      child { node [s_gra] {$0$} 
        child { node [s_gra] {$0$} 
          edge from parent node[left] {$0$} }
        child { node [s_gra] {$0$} 
          edge from parent node[right] {$1$} }
        edge from parent node[left] {$0$} }
      child { node [s_gra] {$0$} 
        child { node [s_gra] {$0$} 
          edge from parent node[left] {$0$} }
        child { node [s_gra] {$1$} 
          edge from parent node[right] {$1$} }
        edge from parent node[right] {$1$} }
      edge from parent node[above] {$0$} } 
    child { node [s_gra] {$1$} 
      child { node [s_gra] {$1$} 
        child { node [s_gra] {$1$} 
          edge from parent node[left] {$0$} }
        child { node [s_gra] {$1$} 
          edge from parent node[right] {$1$} }
        edge from parent node[left] {$0$} }
      child { node [s_gra] {$2$} 
        child { node [s_gra] {$2$} 
          edge from parent node[left] {$0$} }
        child { node [s_gra] {$3$} 
          edge from parent node[right] {$1$} }
        edge from parent node[right] {$1$} }
      edge from parent node[above] {$1$} } 
    edge from parent node[right] {$1$}
  }
;
\end{tikzpicture}
\end{minipage}
\caption{The first levels of the corresponding $3$-linear tree.}
    \label{fig:ex1b}
  \end{figure}
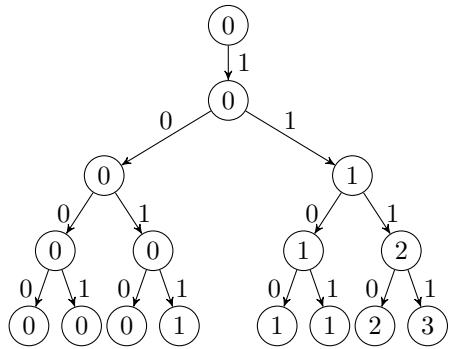
\end{example}

\begin{remark}\label{rk: factors of height h}
  In the tree $T(L_k)$ corresponding to the base-$k$ numeration system, there are exactly two factors of height $h$: the prefix of height $h$ of the tree (occurring only once and whose root has $k-1$ children) and the full $k$-ary tree of height $h$ (occurring everywhere else and whose root has $k$ children). This is consistent with Lemma~\ref{lem:auto} and Remark~\ref{rem:isoTw}: the minimal automaton of $L_k$ has two states. In the initial state, one cannot read leading zeroes. The second state has a loop for any symbol in $A_k$ including zero.
\end{remark}

\begin{example}[Fibonacci numeration system]\label{exa:fib}
  Assume that the decorations take values in $\mathbb{N}$ and that we consider the tree $T(F)$ in the Fibonacci numeration system  with the language $F=1\{0,01\}^*\cup\{\varepsilon\}$ over $\{0,1\}$. Compared with the previous examples, we have two non-isomorphic (undecorated) factors $t_1,t_2$ of height~$2$ to consider as shown in Figure~\ref{fig:ex2}. If we proceed as in the first example, from the prefix of height $2$ and the linear combinations given in Figure~\ref{fig:ex2}, we can build a decorated infinite $2$-linear tree. Start with the prefix denoted by $t_0$. The factor $T[1,2]$ is of type $t_1$. Thus knowing the decorations of the first two levels, we compute the decorations on level~$3$ using the relations in the second tree in Figure~\ref{fig:ex2}. Now, the factor $T[10,2]$ is of type $t_2$ and we compute the decorations on level~$4$ using the relations in the third tree in Figure~\ref{fig:ex2}. Then we have to consider the two factors $T[100,2]$ and $T[101,2]$ of types $t_2$ and $t_1$ respectively and we make use of the relations to compute the next decorations.
   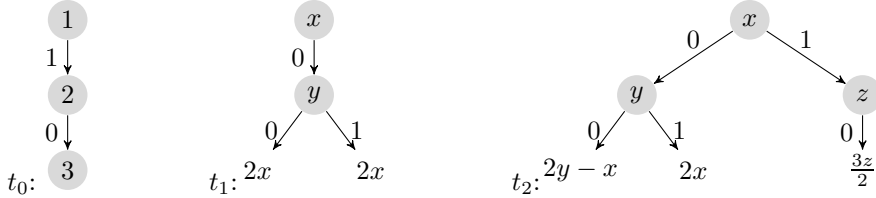
\begin{figure}[h!tb]
    \centering
\tikzset{
  s_gra/.style = {circle,fill=gray!30, inner sep=1pt, minimum size=15pt},
}
$t_0$:\begin{tikzpicture}[->,>=stealth',level/.style={sibling distance = 4cm/#1},level distance = 1cm]
  \node [s_gra] {$1$} 
  child { node [s_gra] {$2$} 
    child { node [s_gra] {$3$} edge from parent node[left] {$0$}}
edge from parent node[left] {$1$}}  
;
\end{tikzpicture}\hskip1.5cm
\tikzset{
  s_gra/.style = {circle,fill=gray!30, inner sep=1pt, minimum size=15pt},
  s_whi/.style = {inner sep=1pt, minimum size=15pt}
}
$t_1$:\begin{tikzpicture}[->,>=stealth',level/.style={sibling distance = 3cm/#1},level distance = 1cm]
  \node [s_gra] {$x$} 
  child { node [s_gra] {$y$} 
    child { node [s_whi] {$2x$} edge from parent node[left] {$0$}}
    child { node [s_whi] {$2x$} edge from parent node[right] {$1$}}
    edge from parent node[left] {$0$}}  
;
\end{tikzpicture}\hskip1.5cm
\tikzset{
    s_gra/.style = {circle,fill=gray!30, inner sep=1pt, minimum size=15pt},
  s_whi/.style = {inner sep=1pt, minimum size=15pt}
}
$t_2$:\begin{tikzpicture}[->,>=stealth',level/.style={sibling distance = 3cm/#1},level distance = 1cm]
  \node [s_gra] {$x$} 
  child { node [s_gra] {$y$} 
    child { node [s_whi] {$2y-x$} edge from parent node[left] {$0$}}
    child { node [s_whi] {$2x$} edge from parent node[right] {$1$}}
    edge from parent node[above] {$0$}}
   child { node [s_gra] {$z$} 
    child { node [s_whi] {$\frac{3z}{2}$} edge from parent node[left] {$0$}}
edge from parent node[above] {$1$}}  
;
\end{tikzpicture}
    \caption{A second example (Fibonacci case): the prefix of height $2$ and linear relations to extend the tree.}
    \label{fig:ex2}
  \end{figure}
Keeping doing so, we obtain the tree in Figure~\ref{fig:ex2b} and the serialization of the decorations is the sequence $1,2,3,4,4,5,6,6,6,8,9,8,8,7,10,12,12,12,10,12,12,\ldots$ {\cite[A282717]{OEIS}} counting the number of nonzero entries in $n$th row of the generalized Pascal triangle based on Fibonacci representations. The linear relations are taken from \cite{LRS}.
  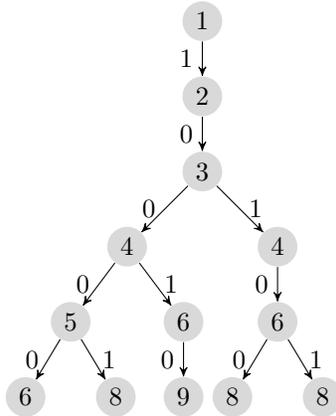
\begin{figure}[h!tb]
    \centering
\tikzset{
  s_gra/.style = {circle,fill=gray!30, inner sep=1pt, minimum size=15pt}
}
\begin{tikzpicture}[->,>=stealth',level/.style={sibling distance = 6cm/#1},level distance = 1cm]
  \node [s_gra] {$1$} 
  child { node [s_gra] {$2$} 
    child { node [s_gra] {$3$} 
      child {node [s_gra] {$4$} 
        child {node [s_gra] {$5$} 
          child {node [s_gra] {$6$} edge from parent node[left] {$0$}}
          child {node [s_gra] {$8$} edge from parent node[right] {$1$}}
        edge from parent node[left] {$0$}}
      child {node [s_gra] {$6$}
child {node [s_gra] {$9$} edge from parent node[left] {$0$}}
        edge from parent node[right] {$1$}} 
            edge from parent node[left] {$0$}}
    child {node [s_gra] {$4$} 
      child {node [s_gra] {$6$}
        child {node [s_gra] {$8$} edge from parent node[left] {$0$}}
        child {node [s_gra] {$8$} edge from parent node[right] {$1$}}
        edge from parent node[left] {$0$}} 
  edge from parent node[right] {$1$}}
edge from parent node[left] {$0$}}
edge from parent node[left] {$1$}}  
;
\end{tikzpicture}
\caption{The first levels of the corresponding $2$-linear tree.}
    \label{fig:ex2b}
  \end{figure}
\end{example}

Rational bases numeration systems were introduced in \cite{Akiyama--Frougny-Sakarovitch-2008} and are special kinds of abstract numeration systems. Let $p$ and $q$ be two relatively prime integers with $p > q > 1$.
We let $L_\frac{p}{q} \subseteq A_p^*$ denote the numeration language in base $\frac{p}{q}$.
It is known that $L_\frac{p}{q}$ is prefix-closed and non-regular.
The corresponding tree $T(L_\frac{p}{q})$ has a purely periodic labeled signature denoted by $(w_0,\ldots,w_{q-1})$, where each word $w_i$ belongs to $A_p^*$. For details, we refer the reader to \cite{Akiyama--Frougny-Sakarovitch-2008,Marsault-these,Marsault--Sakarovitch-2,Marsault--Sakarovitch-2016,RS}.

\begin{example}[Rational base]\label{exa:rat32}
  Consider the rational base $\frac{3}{2}$ having $(02,1)^\omega$ as signature and the tree $T(L_\frac32)$. Let us quickly recall how this tree is built. Following the breadth-first traversal, each node periodically has either two children with edges of labels $0$ and $2$, or one child with a single edge of label $1$. Except for the root which has only a single edge with label $2$ to avoid representations with leading zeroes, we alternate between nodes of degree one or two. 

  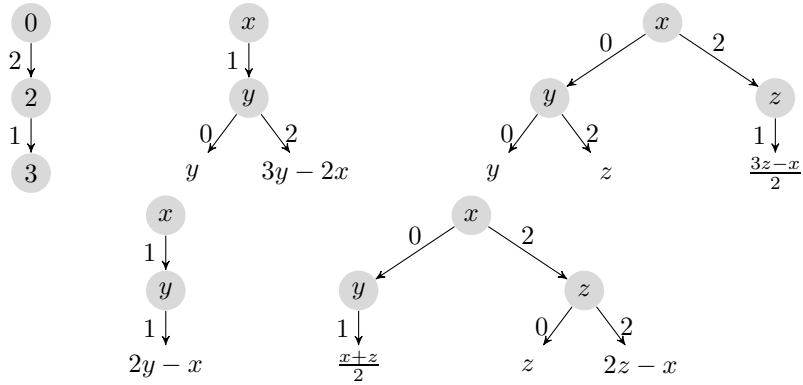
\begin{figure}[h!tb]
    \centering
\tikzset{
  s_gra/.style = {circle,fill=gray!30, inner sep=1pt, minimum size=15pt},
}
\begin{tikzpicture}[->,>=stealth',level/.style={sibling distance = 4cm/#1},level distance = 1cm]
  \node [s_gra] {$0$} 
  child { node [s_gra] {$2$} 
    child { node [s_gra] {$3$} edge from parent node[left] {$1$}}
edge from parent node[left] {$2$}}  
;
\end{tikzpicture}\hskip1.5cm
\tikzset{
  s_whi/.style = {inner sep=1pt, minimum size=15pt},
  s_gra/.style = {circle,fill=gray!30, inner sep=1pt, minimum size=15pt}
}
\begin{tikzpicture}[->,>=stealth',level/.style={sibling distance = 3cm/#1},level distance = 1cm]
  \node [s_gra] {$x$} 
  child { node [s_gra] {$y$} 
    child { node [s_whi] {$y$} edge from parent node[left] {$0$}}
    child { node [s_whi] {$3y-2x$} edge from parent node[right] {$2$}}
    edge from parent node[left] {$1$}}  
;
\end{tikzpicture}\hskip1.5cm
\tikzset{
  s_whi/.style = {inner sep=1pt, minimum size=15pt},
  s_gra/.style = {circle,fill=gray!30, inner sep=1pt, minimum size=15pt}
}
\begin{tikzpicture}[->,>=stealth',level/.style={sibling distance = 3cm/#1},level distance = 1cm]
  \node [s_gra] {$x$} 
  child { node [s_gra] {$y$} 
    child { node [s_whi] {$y$} edge from parent node[left] {$0$}}
    child { node [s_whi] {$z$} edge from parent node[right] {$2$}}
    edge from parent node[above] {$0$}}
   child { node [s_gra] {$z$} 
    child { node [s_whi] {$\frac{3z-x}{2}$} edge from parent node[left] {$1$}}
edge from parent node[above] {$2$}}  
;
\end{tikzpicture}\\
\tikzset{
  s_whi/.style = {inner sep=1pt, minimum size=15pt},
  s_gra/.style = {circle,fill=gray!30, inner sep=1pt, minimum size=15pt}
}
\begin{tikzpicture}[->,>=stealth',level/.style={sibling distance = 3cm/#1},level distance = 1cm]
  \node [s_gra] {$x$} 
  child { node [s_gra] {$y$} 
    child { node [s_whi] {$2y-x$} edge from parent node[left] {$1$}}
    edge from parent node[left] {$1$}}  
;
\end{tikzpicture}\hskip1.5cm
\tikzset{
  s_whi/.style = {inner sep=1pt, minimum size=15pt},
  s_gra/.style = {circle,fill=gray!30, inner sep=1pt, minimum size=15pt}
}
\begin{tikzpicture}[->,>=stealth',level/.style={sibling distance = 3cm/#1},level distance = 1cm]
  \node [s_gra] {$x$} 
  child { node [s_gra] {$y$} 
    child { node [s_whi] {$\frac{x+z}{2}$} edge from parent node[left] {$1$}}
    edge from parent node[above] {$0$}}
   child { node [s_gra] {$z$} 
    child { node [s_whi] {$z$} edge from parent node[left] {$0$}}
    child { node [s_whi] {$2z-x$} edge from parent node[right] {$2$}}
edge from parent node[above] {$2$}}  
;
\end{tikzpicture}
\caption{A third example (base-$\frac{3}{2}$ case).}
    \label{fig:ex3}
  \end{figure}
  
    As in the previous examples, from the prefix of height $2$ and the linear combinations given in Figure~\ref{fig:ex3}, we can build a decorated infinite 2-linear tree on $T(L_\frac32)$. It is depicted in Figure~\ref{fig:ex3b}. We proceed as in the other examples: in $T(L_\frac32)$, if in a given factor the first two levels have been decorated, then we get the decorations on the lower level.
      We get the serialization: $0,2,3,3,5,4,5,7,5,5,7,8,5,7,6,7,9,9,\ldots$, which is the sum-of-digits sequence in base $\frac{3}{2}$. To show that it is indeed this sequence, for the first (respectively second) type of tree, observe that $y-x$ (respectively $z-x$) is equal to $1$ (respectively $2$). One can therefore deduce the relations to be considered.
  For instance, in the second tree, the second leaf has decoration $3y-2x=y+2 \cdot (y-x)=y+2$, which indeed reflects the sum of digits.

  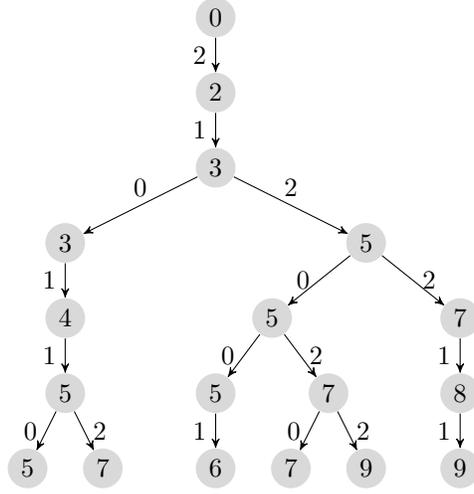
\begin{figure}[h!t]
    \centering
    \tikzset{
  s_gra/.style = {circle,fill=gray!30, inner sep=1pt, minimum size=15pt}
}
\tikzstyle{level 1}=[sibling distance=60mm]
\tikzstyle{level 2}=[sibling distance=50mm]
\tikzstyle{level 3}=[sibling distance=40mm]
\tikzstyle{level 4}=[sibling distance=25mm]
\tikzstyle{level 5}=[sibling distance=15mm]
\tikzstyle{level 6}=[sibling distance=10mm]
\begin{tikzpicture}[->,>=stealth',level distance = 1cm]
  \node [s_gra] {$0$} 
  child { node [s_gra] {$2$} 
    child { node [s_gra] {$3$} 
      child {node [s_gra] {$3$} 
        child {node [s_gra] {$4$}
          child {node [s_gra] {$5$}
            child {node [s_gra] {$5$} edge from parent node[left] {$0$}}
            child {node [s_gra] {$7$} edge from parent node[right] {$2$}}
            edge from parent node[left] {$1$}}
          edge from parent node[left] {$1$}} 
            edge from parent node[above] {$0$}}
    child {node [s_gra] {$5$} 
      child {node [s_gra] {$5$}
        child {node [s_gra] {$5$}
          child {node [s_gra] {$6$} edge from parent node[left] {$1$}}
          edge from parent node[left] {$0$}} 
        child {node [s_gra] {$7$}
          child {node [s_gra] {$7$} edge from parent node[left] {$0$}}
          child {node [s_gra] {$9$} edge from parent node[right] {$2$}}
          edge from parent node[right] {$2$}} 
        edge from parent node[left] {$0$}} 
      child {node [s_gra] {$7$}
        child {node [s_gra] {$8$}
          child {node [s_gra] {$9$} edge from parent node[left] {$1$}}
          edge from parent node[left] {$1$}}
        edge from parent node[right] {$2$}} 
  edge from parent node[above] {$2$}}
edge from parent node[left] {$1$}}
edge from parent node[left] {$2$}}  
;
\end{tikzpicture}
    \caption{The first levels of the corresponding $2$-linear tree.}
    \label{fig:ex3b}
  \end{figure}
\end{example}

\section{Link with regular sequences, so many kernels}\label{sec: regularite et arbres}

Let $k\ge 2$ be an integer. We can easily characterize $k$-regular sequences in terms of the linearity of the tree they decorate.  Recall that the $k$-kernel of a sequence $\mathbf{x}=x_0x_1 \cdots$ is the set
\begin{equation}
  \label{eq:kker}
\ker_k(\mathbf{x}):=\{(x_{k^jn+r})_{n\ge 0}\mid j\ge 0, 0\le r<k^j\}.  
\end{equation}
The $\mathbb{L}$-module generated by a subset $B$ is defined as the set $\langle B\rangle_\mathbb{L}$ of finite linear combinations of elements in $B$ with coefficients in $\mathbb{L}$ of the form $\ell_1 b_1+\cdots +\ell_m b_m$ for some $m\ge 0$ where $\ell_i \in \mathbb{L}$ and $b_i \in B$.
A sequence $\mathbf{x}\in\mathbb{K}^\mathbb{N}$ is said to be {\em $(\mathbb{L},k)$-regular} if the $\mathbb{L}$-module $\langle\ker_k(\mathbf{x})\rangle_\mathbb{L}$ generated by the $k$-kernel of the sequence $\mathbf{x}$ is finitely generated \cite{Ring}.

\begin{remark}\label{rem:gen}
Let us recall an important fact that we will use several times in the following.
  As observed in \cite[Theorem~2.2]{Ring}, if $\langle A\rangle_\mathbb{L}= \langle \mathbf{s}_1,\ldots, \mathbf{s}_m\rangle_\mathbb{L}$, then each $\mathbf{s}_i$ is a $\mathbb{L}$-linear combination of some elements of $A$. Since there is a finite number of $\mathbf{s}_i$ to consider, these elements $\mathbf{s}_1,\ldots, \mathbf{s}_m$ are in fact generated by a finite number of elements of the set $A$. This means that $\langle A\rangle_\mathbb{L}$ is generated by a finite number of elements in $A$. We will make use of this remark first with $A=\ker_k(\mathbf{x})$.
\end{remark}

\begin{theorem}\label{the:kautom}
Let $k\ge 2$ be an integer.  A sequence $\mathbf{x}=x_0x_1 \cdots$ taking values in $\mathbb{K}$ is $(\mathbb{L},k)$-regular if and only if the decorated tree $T_\mathbf{x}(L_k)$ is $(\mathbb{L},h)$-linear for some $h\ge 1$.
\end{theorem}

\begin{proof}
  If the $\mathbb{L}$-module $M$ generated by the $k$-kernel is finitely generated, then there exists some integer $h\ge 1$ such that, for all $0\le r<k^h$, the sequence $(x_{k^hn+r})_{n\ge 0}$ is a linear combination of sequences of the form $(x_{k^jn+s})_{n\ge 0}$ for $0\le j<h$ and $0\le s<k^j$. Using Remark~\ref{rem:gen}, we can assume that $M$ is generated by finitely many elements of the $k$-kernel itself. Otherwise stated, there exist constants $c_{r,j,s}\in\mathbb{L}$ such that 
  \begin{equation}
    \label{eq:khlin}
    \forall n\ge 0,\quad x_{k^hn+r}=\sum_{j=0}^{h-1}\sum_{s=0}^{k^j-1} c_{r,j,s}\, x_{k^jn+s}.
  \end{equation}
For the base-$k$ numeration system, the language $L_k$ and a word $w\in L_k$, the tree $T[w,h]$ of height $h$ is such that
$$\val_k(\dom(T[w,h]))= \{ k^j \val_k(w)+ s \mid 0\le j\le h, 0\le s<k^j\}.$$
So $T_\mathbf{x}(L_k)$ is $(\mathbb{L},h)$-linear. Indeed, Equation~\eqref{eq:khlin} means that in every tree $T[w,h]$ for $w\in L_k$, the decorations on the last level satisfy the same linear relations depending on decorations on the other levels. 

Conversely, if $T_\mathbf{x}(L_k)$ is $(\mathbb{L},h)$-linear, then we obtain relations similar to \eqref{eq:khlin} for all $n\ge 1$. (Indeed, the case of the prefix of height $h$ may give other relations for $n=0$ because $T[\varepsilon,h]$ and $T[w,h]$, for $w\neq\varepsilon$, are not equal as undecorated trees. This is due to the absence of leading zeroes in greedy representations. Recall Remark~\ref{rk: factors of height h}.) Let $0\le r<k^h$. Hence
$(x_{k^hn+r})_{n\ge 1}$ is a linear combination of the sequences $(x_{k^jn+s})_{n\ge 1}$ for $0\le j <h$ and $0\le s<k^j$. Now $(x_{k^hn+r})_{n\ge 0}$ is a linear combination of the sequences $(x_{k^jn+s})_{n\ge 0}$ for $0\le j<h$ and $0\le s<k^j$ and some extra sequence $\mathbf{z}$ whose terms are all zeroes except the first one. From this, we deduce that the sequence $\mathbf{x}$ is $(\mathbb{L},k)$-regular.
Indeed, for $t\ge 1$ and $0\le p< k^{h+t}$, if we write $p=k^h q_p + r_p$, then we have
  \begin{equation}
    \label{eq:iterk}
    (x_{k^{h+t}n+p})_{n\ge 1}=(x_{k^h(k^tn+q_p)+r_p})_{n\ge 1}.
  \end{equation}
  We can therefore express the former sequence as a linear combination of sequences of the form $(x_{k^j(k^t n+q_p)+s})_{n\ge 1}$ for $0\le j <h$ and $0\le s<k^j$ and continue recursively until we get a combination of the sequences $(x_{k^jn+s})_{n\ge 1}$ for $0\le j <h$ and $0\le s<k^j$.
  Therefore the sequence $ (x_{k^{h+t}n+p})_{n\ge 0}$ is a linear combination of sequences $(x_{k^jn+s})_{n\ge 1}$, for $0\le j <h$ and $0\le s<k^j$, and $\mathbf{z}$.
\end{proof}

We consider the definition from \cite{CCS} extending the one of \cite{MR} which permits us to consider $S$-regular sequences in particular abstract numerations systems. 
Let $S=(L,A,<)$ be an abstract numeration system built on a prefix-closed regular language and let $\mathbf{x}=x_0x_1\cdots$ be a sequence in $\mathbb{K}^\mathbb{N}$. 
For all words $u\in A^*$, consider the sequence 
$$\tau(\mathbf{x},u):n\mapsto\left\{
  \begin{array}{ll}
    x_{\rep_S(n)u}, & \text{ if } \rep_S(n)u\in L;\\
    0, & \text{ otherwise}.\\
  \end{array}\right.
  $$
  The {\em $S$-kernel} of the sequence $\mathbf{x}$ is the set
  \begin{equation}
    \label{eq:sker}
    \ker_S(\mathbf{x}):=\{ \tau(\mathbf{x},u) \mid u\in A^*\}.
  \end{equation}
  \begin{remark}[About leading zeroes]
    For a sequence $\mathbf{x}$, there is a small difference between the elements of its $k$-kernel and those of its $S_k$-kernel for the abstract numeration system $S_k$ built on the language $L_k$ from \eqref{eq:Lk}. This difference arises because leading zeroes are allowed in the positional base-$k$ numeration system but are forbidden in the abstract numeration system setting (leading zeroes change the length and thus the value for the radix order). Indeed, when a word $u$ has digit $0$ as a prefix, then $\rep_S(0)u=\varepsilon u=u$ starts with $0$ and does not belong to $L_k$.
    So the sequence $\tau(\mathbf{x},u)$ in $\ker_{S_k}(\mathbf{x})$ starts with $0$ compared with the sequence $(\mathbf{x}(k^{|u|}n+\val_k(u)))_{n\ge 0}$ in $\ker_k(\mathbf{x})$ starting with $\mathbf{x}(k^{|u|}0+\val_k(u))=x_{\val_k(u)}$, which can be non-zero. So the two subsequences built from the same suffix $u$ can differ on the first term. With this distinction, this is why the prefix of height~$h$ in $T(L_k)$ differs from any other factor having a leftmost branch labeled with zeroes. 
    Nevertheless from \cite{CCS}, $\langle\ker_k(\mathbf{x})\rangle_\mathbb{L}$ is finitely generated if and only if $\langle\ker_{S_k}(\mathbf{x})\rangle_\mathbb{L}$ is.
  \end{remark}
  
\begin{definition}
Let $S=(L,A,<)$ be an abstract numeration system built on a prefix-closed regular language $L$.
  A sequence $\mathbf{x}$ taking values in $\mathbb{K}$ is {\em $(\mathbb{L},S)$-regular} if the $\mathbb{L}$-module $\langle\ker_S(\mathbf{x})\rangle_\mathbb{L}$ generated by the $S$-kernel of $\mathbf{x}$ is finitely generated.
\end{definition}

It will be convenient to convey an extra piece of information about the type of tree that is encountered when reading $\rep_S(n)$ in the process of building subsequences. 
Let $h\ge 1$. There are finitely many (undecorated) labeled trees $\{T_1,\ldots,T_{r(h)}\}$ of height $h$ occurring in $T(L)$. For all words $u\in A^*$ and all $j\in\{1,\ldots,r(h)\}$, consider the sequence
$$\tau(\mathbf{x},u,T_j):n\mapsto\left\{
  \begin{array}{ll}
    x_{\rep_S(n)u}, & \text{ if } \rep_S(n)u\in L \text{ and } T[\rep_S(n),h]=T_j;\\
    0, & \text{ otherwise}.\\
  \end{array}\right.
$$
In particular, we have
\begin{equation}
  \label{eq:split}
  \sum_{j=1}^{r(h)} \tau(\mathbf{x},u,T_j) = \tau(\mathbf{x},u).
\end{equation}
The {\em $h$-filtered $S$-kernel} of the sequence $\mathbf{x}$ is the set
  $$\ker_{h,S}(\mathbf{x}):=\{ \tau(\mathbf{x},u,T_j) \mid u\in A^*, 1\le j\le r(h)\}.$$  
Note that the $0$-filtered $S$-kernel is simply the $S$-kernel of the sequence.

For convenience, we introduce some characteristic sequences, for $j\in\{1,\ldots,r(h)\}$, 
$$\chi_j:n\mapsto\left\{
  \begin{array}{ll}
    1, & \text{ if } T[\rep_S(n),h]=T_j;\\
    0, & \text{ otherwise}.\\
  \end{array}\right.
  $$
  This means that
  $$\tau(\mathbf{x},u) \odot \chi_j = \tau(\mathbf{x},u,T_j),$$
  where $\odot$ is the term-wise multiplication (which is compatible with the notation of the Hadamard product, see Remark~\ref{rem:hadamard}).
  \begin{example}
    Consider the three trees $t_0,t_1,t_2$ from Figure~\ref{fig:ex2}. The tree $t_0$ has $\varepsilon$ as root. Trees of type $t_1$ have a root reached with a path ending with $1$ (after a letter $1$, one can only add a letter $0$). Trees of type $t_2$ have a root reached with a path ending with $0$ (after a letter $0$, one can add either $0$ or $1$).
In the next table, we give the values of $\tau(\mathbf{x},0,t_1)$ and $\tau(\mathbf{x},0,t_2)$. Note that $\tau(\mathbf{x},0,t_0)=0,0,0,...$.
  $$\begin{array}{r|rrrrrrrrrrr}
  \rep_S(n) & \varepsilon & 1 & 10 & 100 & 101 & 1000 & 1001 & 1010 & 10000 \\
      \rep_S(n) 0 & 0 & 10 & 100 & 1000 & 1010 & 10000 & 10010 & 10100 & 100000 \\
      \hline
          \tau(\mathbf{x},0)     & 0 & 3 & 4 & 5 & 6 & 6 & 9 & 8 & 7 \\
           \chi_1 & 0&1&0&0&1&0&1&0&0\\
  \tau(\mathbf{x},0,t_1)     & 0 & 3 & 0 & 0 & 6 & 0 & 9 & 0 & 0 \\
      \chi_2 & 0&0&1&1&0&1&0&1&1\\
      \tau(\mathbf{x},0,t_2)      & 0 & 0 & 4 & 5 & 0 & 6 & 0 & 8 & 7 \\
   \end{array}$$
  
\end{example}

  \begin{lemma}
    Let $h\ge 0$. The $S$-kernel of a sequence is finite if and only if its $h$-filtered kernel is finite.
  \end{lemma}

  \begin{proof}
    Each element $\tau(\mathbf{x},u)$ of the $S$-kernel gives rise to at most $r(h)$ distinct elements $\tau(\mathbf{x},u,T_j)$ of the $h$-filtered kernel. Conversely, from Equation~\eqref{eq:split}, each element of the $S$-kernel is obtained by summing up some of the elements of the $h$-filtered kernel.
  \end{proof}
  
  \begin{lemma}\label{lem:fg}
    The $\mathbb{L}$-module $\langle\ker_S(\mathbf{x})\rangle_\mathbb{L}$ generated by the $S$-kernel of $\mathbf{x}$ is finitely generated if and only if the $\mathbb{L}$-module $\langle\ker_{h,S}(\mathbf{x})\rangle_\mathbb{L}$ generated by the $h$-filtered $S$-kernel of $\mathbf{x}$ is finitely generated for all $h\ge 0$.
  \end{lemma}

  \begin{proof}
    From Equation~\eqref{eq:split}, we deduce that $\langle\ker_{S}(\mathbf{x})\rangle_\mathbb{L}$ is included in $\langle\ker_{h,S}(\mathbf{x})\rangle_\mathbb{L}$. 
  Any submodule of a finitely generated module over a Noetherian ring is also
  finitely generated. Hence, if $\langle\ker_{h,S}(\mathbf{x})\rangle_\mathbb{L}$ is finitely generated, so is $\langle\ker_S(\mathbf{x})\rangle_\mathbb{L}$.

  Assume that $\langle\ker_S(\mathbf{x})\rangle_\mathbb{L}$ is finitely generated. Using Remark~\ref{rem:gen}, we may assume that it is generated by a finite number of elements $\mathbf{s}_1,\ldots,\mathbf{s}_m\in \ker_S(\mathbf{x})$. Every element of $\ker_{h,S}(\mathbf{x})$ is of the form $\mathbf{t}\odot \chi_j$ for some $\mathbf{t}\in \ker_S(\mathbf{x})$ and some $1\le j\le r(h)$. Hence, it is a $\mathbb{L}$-linear combination of $\mathbf{s}_1\odot\chi_j,\ldots,\mathbf{s}_m\odot\chi_j$. Otherwise stated, this means that the $h$-filtered $S$-kernel is generated by the sequences $\mathbf{s}_i\odot\chi_j$, $1\le i\le m$, $1\le j\le r(h)$.
\end{proof}

\begin{prop}
Let $S=(L,A,<)$ be an abstract numeration system built on a prefix-closed (not necessarily regular) language.  If the $\mathbb{L}$-module $\langle\ker_{S}(\mathbf{x})\rangle_\mathbb{L}$ generated by the $S$-kernel of a sequence $\mathbf{x}$ taking values in $\mathbb{K}$ is finitely generated, then the decorated tree $T_\mathbf{x}(L)$ is $(\mathbb{L},h)$-linear for some $h\ge 1$.
\end{prop}

\begin{proof}
  As in the proof of the above lemma, we may assume that $\langle\ker_{S}(\mathbf{x})\rangle_\mathbb{L}$ is generated by a finite number of sequences $\tau(\mathbf{x},u_1),\ldots,\tau(\mathbf{x},u_m)$. Let $h=1+\max_i |u_i|$. By Lemma~\ref{lem:fg}, $\langle\ker_{h,S}(\mathbf{x})\rangle_\mathbb{L}$ is generated by the sequences of the form $\tau(\mathbf{x},u_i,T_j)=\tau(\mathbf{x},u_i)\odot\chi_j$, $i=1,\ldots,m$, $j=1,\ldots,r(h)$. Let $v$ be a word of length $h$ and $j$ be such that $v$ belongs to the domain of $T_j$ (one of the labeled factors of height $h$). Then, $\tau(\mathbf{x},v,T_j)$ is a $\mathbb{L}$-linear combination of the $\tau(\mathbf{x},u_i,T_j)$. Due to the form of the sequences (obtained as a multiplication by $\chi_j$), we can indeed assume that no sequence of the form $\tau(\mathbf{x},u',T_i)$ with $i\neq j$ occurs in the decomposition of $\tau(\mathbf{x},v,T_j)$. For any $n$ such that $\rep_S(n)$ is the root of a factor of the type $T_j$, this means that $x_{\rep_S(n)v}$ is a $\mathbb{L}$-linear combination of the $x_{\rep_S(n)u_i}$'s. Since $|v|=h$, this means that in any factor of height $h$, the decorations of the leaves are linear combinations of decorations of the above nodes (since $|u_i|\le h-1$). We conclude that $T_\mathbf{x}(L)$ is $(\mathbb{L},h)$-linear. 
\end{proof}

In the next proposition, the assumption that the numeration language is regular is important. Also see Remark~\ref{rk: regular lang important assumption}.

\begin{prop}
  Let $S=(L,A,<)$ be an abstract numeration system built on a prefix-closed regular language.  If the  tree $T_\mathbf{x}(L)$ decorated by a sequence $\mathbf{x}$ taking values in $\mathbb{K}$ is $(\mathbb{L},h)$-linear for some $h\ge 1$, then the sequence $\mathbf{x}$ is $(\mathbb{L},S)$-regular.
\end{prop}

\begin{proof}
    By Lemma~\ref{lem:h+1}, we can take $h$ large enough such that the factors of height $h$ in $T(L)$ are in one-to-one correspondence with the states of the minimal automaton of $L$ (see Remark~\ref{rem:isoTw}) and this correspondence is compatible with the transition function of the automaton, i.e., if a state $q$ corresponds to $T[w,h]$, then $q.a$ corresponds to $T[wa,h]$. So instead of denoting by $T_1,\ldots,T_{r(h)}$ the trees of height $h$ occurring in $T(L)$ we can denote them by their respective root $T_w$.
  
By assumption, $T_\mathbf{x}(L)$ is $(\mathbb{L},h)$-linear for some $h\ge 1$: if $w$ is the root of a tree $T[w,h]$, then for all word $u$ of length $h$ such that $wu\in L$, there exist constants $c_{w,u,v}\in\mathbb{L}$ such that 
\begin{equation}\label{eq:decomp}
x_{wu} = \sum_{\substack{v\in \dom (T_w)\\ |v|<h}} c_{w,u,v}\, x_{wv}.
\end{equation}
If $u$ is such that $wu\not\in L$, then the constants are set to $0$.
Hence  $\tau(\mathbf{x},u,T_w)$ is a $\mathbb{L}$-linear combination of sequences of the form $\tau(\mathbf{x},v,T_w)$ for $|v|<h$. 

We proceed as in Equation~\eqref{eq:iterk} to show that any element of the $h$-filtered kernel can be expressed as a linear combination of sequences of the form $\tau(\mathbf{x},v,T_w)$ for $|v|<h$. Note that there are finitely many such sequences (as there are finitely many trees $T_w$). Let $a\in A$ such that $au$ is a word of length $h+1$ and $wau\in L$. By using Equation~\eqref{eq:decomp} in the first and last equalities, we have 
  \begin{eqnarray*}
    x_{wau} &=& \sum_{\substack{v\in \dom (T_{wa})\\ |v|<h}} c_{wa,u,v}\, x_{wav} \\
                     &=& \sum_{\substack{v\in \dom (T_{wa})\\ |v|<h-1}} c_{wa,u,v}\, x_{wav} + \sum_{\substack{v\in \dom (T_{wa})\\ |v|=h-1}} c_{wa,u,v}\, x_{wav} \\
    &=& \sum_{\substack{v\in \dom (T_{wa})\\ |v|<h-1}} c_{wa,u,v}\, x_{wav} + \sum_{\substack{v\in \dom (T_{wa})\\ |v|=h-1}} c_{wa,u,v}\, \sum_{\substack{z\in \dom (T_w)\\ |z|<h}} c_{w,av,z}\, x_{wz},
  \end{eqnarray*}
  which means that $\tau(\mathbf{x},au,T_w)$ is a $\mathbb{L}$-linear combination of sequences of the form $\tau(\mathbf{x},av,T_{wa})$ for $|av|<h$ and of sequences of the form $\tau(\mathbf{x},z,T_{w})$ for $|z|<h$. Repeating this argument, this shows that the $h$-filtered $S$-kernel is finitely generated by elements of the form $\tau(\mathbf{x},v,T_w)$ for $|v|<h$ and $w\in A^*$. Hence we conclude with Lemma~\ref{lem:fg}.
\end{proof}

\begin{remark}\label{rk: regular lang important assumption}
  We stress out that the assumption on the regularity of the language $L$ appears quite subtly in the above proof. Assuming that $h$ is large enough, when we have a factor $T[w,h]$ occurring in $T(L)$, adding an extra letter $a$ leads to a unique tree $T[wa,h]$. There is no ambiguity because there is an underlying finite automaton recognizing $L$. More precisely, whenever $T[w,h]=T[w',h]$, then $T[wa,h]=T[w'a,h]$. See Lemma~\ref{lem:auto} and Remark~\ref{rem:isoTw}.
  
  In a more general setting such as rational base numeration systems, this is no longer the case. In base $\frac{3}{2}$ for instance, for any height $h$ and any word $w$, the factor $T[w,h]$ appears as the prefix of two distinct trees of height $h+1$. 
For example, see the four trees in the right side of Figure~\ref{fig:ex3}: each tree of height $1$ is the prefix of two distinct trees of height $2$. So $T[w,h]=T[w',h]$ never implies that $T[wa,h]=T[w'a,h]$.

\end{remark}

We summarize the above two propositions as follows.

\begin{theorem}\label{the:Sreg}
  Let $S=(L,A,<)$ be an abstract numeration system built on a prefix-closed regular language. A sequence $\mathbf{x}$ taking values in $\mathbb{K}$ is $(\mathbb{L},S)$-regular sequence if and only if the decorated tree $T_\mathbf{x}(L)$ is $(\mathbb{L},h)$-linear for some $h\ge 1$.
\end{theorem}

\begin{remark}\label{rem:hadamard}
Under the conditions of the above theorem, as a consequence of \cite[Theorem 29]{CCS}, the series $S=\sum_{w\in L} x_w\, w$ is recognizable. Note that since $L$ is a regular language, the series $S_j=\sum_{w\in L} \chi_j(\val_S(w))\, w$ is the characteristic series of a regular language $L_j = \{ w\in L \mid T[w,h]=T_j \}$. Hence, the Hadamard product $S \odot S_j$ of these two series is again recognizable.
\end{remark}

\section{What is a $\frac{p}{q}$-regular sequence\protect\footnote{The idea to start this paper came after listening to a talk ``{\em Avoiding fractional powers on an infinite alphabet}'' given by E. Rowland \cite{RowlandSti}.}?}\label{sec: p/q-regularite}

Let $\mathbf{s}$ denote the sum-of-digits sequence in the base-$\frac32$ numeration system $S=(L_{\frac{3}{2}},\{0,1,2\},<)$ (defined in Section~\ref{sec: intro}).
As mentioned in the introduction, the $3$-kernel of $\mathbf{s}$ given by Equation~\eqref{eq:kker} does not seem to provide any linear relationships.  We give here some more details about the sets of subsequences that can be associated with~$\mathbf{s}$.
The Mathematica package {\sc IntegerSequences} developed by E.~Rowland \cite{IS} implements a procedure for guessing $k$-regular sequences as described by Shallit \cite{Sha}. But it does not find any $2$-, $3$- nor $6$-regularity for the sequence~$\mathbf{s}$.

Let now us have a look at some elements $\tau(\mathbf{s},u)$ of the $S$-kernel of the sequence~$\mathbf{s}$. We will show that $S$-kernel of $\mathbf{s}$ given by Equation~\eqref{eq:sker} is not finitely generated.
$$\begin{array}{r|l}
u\in \{0,1,2\}^* & \tau(\mathbf{s},u) \\
\hline
\varepsilon & 0, 2, 3, 3, 5, 4, 5, 7, 5, 5, 7, 8, 5, 7, 6, 7, 9, 9, 5, 7, 8 ,\ldots \\
0 & 0, 0, 3, 0, 5, 0, 5, 0, 5, 0, 7, 0, 5, 0, 6, 0, 9, 0, 5, 0, 8 ,\ldots \\
1 & 0, 3, 0, 4, 0, 5, 0, 8, 0, 6, 0, 9, 0, 8, 0, 8, 0, 10, 0, 8, 0 ,\ldots \\
10 & 0, 3, 0, 0, 0, 5, 0, 0, 0, 6, 0, 0, 0, 8, 0, 0, 0, 10, 0, 0, 0 ,\ldots \\
01 & 0, 0, 4, 0, 0, 0, 6, 0, 0, 0, 8, 0, 0, 0, 7, 0, 0, 0, 6, 0, 0 ,\ldots \\
11 & 0, 0, 0, 5, 0, 0, 0, 9, 0, 0, 0, 10, 0, 0, 0, 9, 0, 0, 0, 9, 0 ,\ldots \\
00 & 0, 0, 0, 0, 5, 0, 0, 0, 5, 0, 0, 0, 5, 0, 0, 0, 9, 0, 0, 0, 8 ,\ldots \\
     000& 0, 0, 0, 0, 0, 0, 0, 0, 5, 0, 0, 0, 0, 0, 0, 0, 9, 0, 0, 0, 0 ,\ldots \\
          \end{array}
          $$
Note that all suffixes $u$ are not considered in the previous table since some can be obtained from others. For example, \[
\tau(\mathbf{s},2)=2, 0, 5, 0, 7, 0, 7, 0, 7, 0, 9, 0, 7, 0, 8, 0, 11, 0, 7, 0, 10, \ldots
\]
can be expressed as $\tau(\mathbf{s},2)=\tau(\mathbf{s},0)+(20)^\omega$ and $\tau(\mathbf{s},21)=\tau(\mathbf{s},00)+(3000)^\omega$. 
Also observe that, for any $u\in \{0,1,2\}^*$, the sequence $(\tau(\mathbf{s},u)(n))_{n\ge 1}$ has positive values separated by $2^{|u|}-1$ zeroes because there are $2^{|u|}$ types of factors of height~$|u|$ and $u$ belongs to the domain of exactly one of them.

          With basic knowledge about the numeration system $S$, it is not difficult to see that $\ker_S(\mathbf{s})$ contains a sequence having non zero terms only for all positive positions being congruent to $r$ mod $2^j$ for all $j\ge 0$ and all $0\le r<2^j$.
          Since the characteristic sequences of the sets $\mathbb{N}$, $2\mathbb{N}$, $4\mathbb{N}$, $4\mathbb{N}+1$, \ldots, $2^j\mathbb{N}$,\ldots, $2^j\mathbb{N}+2^{j-1}-1$ are linearly independent, then $\langle \ker_S(\mathbf{s})\rangle$ cannot be finitely generated.

As a conclusion, no reasonable kernel associated with $\mathbf{s}$ seems to be finitely generated. In view of Theorem~\ref{the:Sreg}, instead of considering a definition based on the kernel, we therefore propose the following.

          \begin{definition}\label{def:Sregular}
Let $p$ and $q$ be two relatively prime integers with $p > q > 1$.
            A sequence $\mathbf{x}$ taking values in $\mathbb{K}$ is {\em $(\mathbb{L},\frac{p}{q})$-regular} whenever the decorated tree $T_\mathbf{x}(L_{\frac{p}{q}})$ is $(\mathbb{L},h)$-linear for some $h\ge 1$.
          \end{definition}

Due to Theorems~\ref{the:kautom} and~\ref{the:Sreg}, Definition~\ref{def:Sregular} indeed generalizes $k$-regular and $S$-regular sequences in the usual sense.

          We now give some closure properties of $(\mathbb{L},\frac{p}{q})$-regular sequences.

\begin{prop}\label{pro:module}
The set of $(\mathbb{L},\frac{p}{q})$-regular sequences is a $\mathbb{L}$-module.
\end{prop}
\begin{proof}
We have to show that the set of $(\mathbb{L},\frac{p}{q})$-regular is closed under sum and multiplication by a constant.
The case of the multiplication by a constant is clear, so we only prove the case of the sum.
Let $\mathbf{x}$ and $\mathbf{y}$ be two $(\mathbb{L},\frac{p}{q})$-regular sequences.
Assume that the trees $T_\mathbf{x}(L_{\frac{p}{q}})$ and $T_\mathbf{y}(L_{\frac{p}{q}})$ are respectively $h$-linear and $h'$-linear for some $h,h'\ge 1$.
Let $H=\max(h,h')$.
Then both trees are $H$-linear by Lemma~\ref{lem:h+1}, so is the tree $T_{\mathbf{x}+\mathbf{y}}(L_{\frac{p}{q}})$.
Therefore the sequence $\mathbf{x}+\mathbf{y}$ is $(\mathbb{L},\frac{p}{q})$-regular.
\end{proof}

\subsection{Guessing $h$-linearity}
As in \cite[Section~6]{AS2},  a practical procedure may be applied to often succeed in deducing the $\frac{p}{q}$-regularity of a sequence by inspecting its first few terms.

Consider a long enough prefix of the sequence $\mathbf{x}=(x_n)_{n\ge 0}$ of interest. 
If we know or can compute $L$ terms of the sequence, we are thus able to decorate the nodes on the first $\ell$ levels of the tree $T(L_\frac{p}{q})$ for some $\ell \ge 0$. Fix an integer $h\ge 1$. We would like to conjecture the $h$-linearity of this tree by providing convenient candidates for linear relations occurring for the decorations of subtrees of height $h$.

Except for the prefix $T[\varepsilon,h]$, it is well-known that there are $q^h$ non-isomorphic (undecorated) subtrees of height $h$~\cite[Lemme 4.14]{Marsault-these}. More precisely, two factors $T[u,h]$ and $T[v,h]$ are equal if and only if $\val_{\frac{p}{q}}(u)\equiv\val_{\frac{p}{q}}(v)\pmod{q^h}$.

There are $p^h$ linear relations to guess because this is the total number of leaves of the $q^h$ pairwise distinct (undecorated) trees. Let $R\in\{0,\ldots,q^h-1\}$ be a remainder. Consider the finite collection of subtrees $T[u_j,h]$, $j\in J_R$, that can be extracted from the first $\ell$ levels of the tree $T(L_\frac{p}{q})$ and such that $\val_{\frac{p}{q}}(u_j)\equiv R\pmod{q^h}$. These trees are equal as undecorated trees. Let us say that they have $k_R$ leaves and $i_R$ internal nodes. In particular, $k_0+k_1+\cdots+k_{q^h-1}=p^h$. For each $j\in J_R$, we order by breadth-first traversal the nodes of $T[u_j,h]$ and we enumerate the corresponding decorations
$$\underbrace{x_{n_{j,1}},\ldots,x_{n_{j,i_R}}}_{\text{internal nodes}},\underbrace{x_{n_{j,i_R+1}},\ldots,x_{n_{j,i_R+k_R}}}_{\text{leaves}},$$
which is a subsequence of the sequence $(x_n)_{n\ge 0}$. 
For each $j\in J_R$ and for each $t\in\{1,\ldots,k_R\}$, we set up an equation:
$$\alpha_{R,t,1}\, x_{n_{j,1}}+ \cdots +\alpha_{R,t,i_R}\, x_{n_{j,i_R}}= x_{n_{j,i_R+t}},$$ where the coefficients $\alpha_{R,t,i}$ do not depend on $j$.
As $j$ varies in $J_R$, we get a tree (of the same type) with other decorations because the root is different. Assume the $h$-linearity of the decorated tree. If we are interested in the same leaf in every tree, the same linear relation should be satisfied. Therefore, when $j$ varies, we have a system of linear equations whose unknowns are $\alpha_{R,t,1},\ldots,\alpha_{R,t,i_R}$. If $\#J_R\ge i_R$ and the system has no solution, then the tree cannot be $h$-linear and one can test $(h+1)$-linearity. Otherwise, we could conjecture a linear relation expressing the decoration of the considered leaf in terms of the decorations of the internal nodes. Note that we have such a system for each leaf of every one of the $q^h$ non-isomorphic subtrees of height~$h$.

\begin{example}
As an example, consider the sequence of squares $(n^2)_{n\ge 0}$. Let $\frac{p}{q}=\frac{3}{2}$, $h=3$ and $R=2$. For this remainder~$R$, the shape of the tree is given in Figure~\ref{fig:8n+2}.
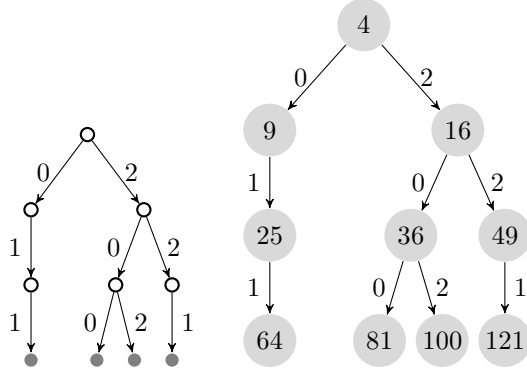
\begin{figure}[h!tb]
  \center
\tikzset{
  s_bla/.style = {circle,fill=white,thick, draw=black, inner sep=0pt, minimum size=5pt},
  s_red/.style = {circle,black,fill=gray,thick, inner sep=0pt, minimum size=5pt}
}
\begin{tikzpicture}[->,>=stealth',level/.style={sibling distance = 1.5cm/#1},level distance = 1cm]
  \node [s_bla] {} 
      child {node [s_bla] {} 
        child {node [s_bla] {} 
          child {node [s_red] {} edge from parent node[left] {$1$} }
          edge from parent node[left] {$1$} }
      edge from parent node[left] {$0$} }
    child {node [s_bla] {} 
      child {node [s_bla] {}
        child {node [s_red] {} edge from parent node[left] {$0$}}
        child {node [s_red] {} edge from parent node[right] {$2$}}
        edge from parent node[left] {$0$}} 
      child {node [s_bla] {}
        child {node [s_red] {} edge from parent node[right] {$1$} }
        edge from parent node[right] {$2$}} 
      edge from parent node[right] {$2$}}
;
\end{tikzpicture}\quad
 \tikzset{
   s_bla/.style = {circle,fill=gray!30, inner sep=1pt, minimum size=20pt},
   s_red/.style = {circle,fill=gray!30, inner sep=1pt, minimum size=20pt}
}
\begin{tikzpicture}[->,>=stealth',level/.style={sibling distance = 2.5cm/#1},level distance = 1.4cm]
  \node [s_bla] {$4$} 
      child {node [s_bla] {$9$} 
        child {node [s_bla] {$25$} 
          child {node [s_red] {$64$} edge from parent node[left] {$1$} }
          edge from parent node[left] {$1$} }
      edge from parent node[left] {$0$} }
    child {node [s_bla] {$16$} 
      child {node [s_bla] {$36$}
        child {node [s_red] {$81$} edge from parent node[left] {$0$}}
        child {node [s_red] {$100$} edge from parent node[right] {$2$}}
        edge from parent node[left] {$0$}} 
      child {node [s_bla] {$49$}
        child {node [s_red] {$121$} edge from parent node[right] {$1$} }
        edge from parent node[right] {$2$}} 
      edge from parent node[right] {$2$}}
;
\end{tikzpicture}
\caption{A type of tree of height $3$ and one occurrence.}\label{fig:8n+2}
\end{figure}
We encounter such a tree with the decorations (in breadth-first traversal) given in Table~\ref{tab:alval} where each line corresponds to an occurrence of such a decorated factor. We have written the  first eight occurrences but we could take more (this would increase the number of equations and thus of possible constraints). Table~\ref{tab:alval} has two columns: the first six values are the decorations of the internal nodes ($i_2=6$), and the last four values are the decorations of the leaves ($k_2=4$).
\begin{table}[h!t]
  \centering
  $$\begin{array}{cccccc|cccc}
      x_{n_j,1} & x_{n_j,2} &x_{n_j,3} &x_{n_j,4} &x_{n_j,5} &x_{n_j,6} &x_{n_j,7} &x_{n_j,8} &x_{n_j,9} &x_{n_j,10} \\
      \hline
 4 & 9 & 16 & 25 & 36 & 49 & 64 & 81 & 100 & 121 \\
 100 & 225 & 256 & 529 & 576 & 625 & 1225 & 1296 & 1369 & 1444 \\
 324 & 729 & 784 & 1681 & 1764 & 1849 & 3844 & 3969 & 4096 & 4225 \\
 676 & 1521 & 1600 & 3481 & 3600 & 3721 & 7921 & 8100 & 8281 & 8464 \\
 1156 & 2601 & 2704 & 5929 & 6084 & 6241 & 13456 & 13689 & 13924 & 14161 \\
    1764 & 3969 & 4096 & 9025 & 9216 & 9409 & 20449 & 20736 & 21025 & 21316 \\
    2500 & 5625 & 5776 & 12769 & 12996 & 13225 & 28900 & 29241 & 29584 & 29929 \\
      3364 & 7569 & 7744 & 17161 & 17424 & 17689 & 38809 & 39204 & 39601 & 40000
  \end{array}$$
  \caption{Decorations of some subtrees.}
  \label{tab:alval}
\end{table}
The goal is to conjecture relations for all four leaves.
If we want to do so for the first leaf, we consider the linear system
$$\left(\begin{array}{cccccc}
 4 & 9 & 16 & 25 & 36 & 49 \\
 100 & 225 & 256 & 529 & 576 & 625 \\
 324 & 729 & 784 & 1681 & 1764 & 1849 \\
 676 & 1521 & 1600 & 3481 & 3600 & 3721 \\
 1156 & 2601 & 2704 & 5929 & 6084 & 6241 \\
 1764 & 3969 & 4096 & 9025 & 9216 & 9409 \\
 2500 & 5625 & 5776 & 12769 & 12996 & 13225 \\
 3364 & 7569 & 7744 & 17161 & 17424 & 17689 \\
        \end{array}\right)
      \begin{pmatrix}
        \alpha_{2,1,1}\\
        \alpha_{2,1,2}\\
        \alpha_{2,1,3}\\
        \alpha_{2,1,4}\\
        \alpha_{2,1,5}\\
        \alpha_{2,1,6}\\
      \end{pmatrix}=
      \begin{pmatrix}
   64 \\
 1225 \\
 3844 \\
 7921 \\
 13456 \\
 20449 \\
 28900 \\
 38809 \\
      \end{pmatrix},
      $$
      where we used the values in the column $x_{n_j,7}$ of Table~\ref{tab:alval} corresponding to the first leaf.
      A solution is given by $ \alpha_{2,1,1}=
        \alpha_{2,1,2}=
        \alpha_{2,1,3}=0$, $\alpha_{2,1,4}=\alpha_{2,1,5}=5/4$, and $\alpha_{2,1,6}=-1/4$.  The solution remains valid when increasing the number of equations (for instance, taking the 20 first occurrences of the subtree). By replacing the column vector in the right-hand side with one of the last tree columns in Table~\ref{tab:alval} (they thus correspond to the last three leaves in Figure~\ref{fig:8n+2}), we get similar solutions and we conjecture the following values for the coefficients: 
\begin{tabular}{l}
        $ \alpha_{2,2,j}=0$ if $j\neq 5$ and $\alpha_{2,2,5}=9/4$;
\\
         $ \alpha_{2,3,1}=
        \alpha_{2,3,2}=
        \alpha_{2,3,3}=0$, $ \alpha_{2,3,4}=-1/4$, and $\alpha_{2,3,5}=\alpha_{2,3,6}=5/4$;
\\
         $ \alpha_{2,4,1}=
        \alpha_{2,4,2}=
        \alpha_{2,4,3}=0$, $ \alpha_{2,4,4}=1/2$, $\alpha_{2,4,5}=-7/4$, and $\alpha_{2,4,6}=7/2$.
\end{tabular}

In order to obtain the $h$-linearity of the tree, we have to do the same with the seven other remainders modulo $8$ and thus seven other types of trees with a total of $27-4=23$ other leaves. Having such conjectures, these relations are easy to prove. Indeed, let us take an example and let us prove the first relation we conjectured above. Recalling that
        \begin{equation}
          \label{eq:numsys}
\val_{\frac{p}{q}}(uv)=\val_{\frac{p}{q}}(u)\left(\frac{p}{q}\right)^{|v|}+\val_{\frac{p}{q}}(v),          
        \end{equation}
if $u$ is the $\frac{3}{2}$-representation of $2+8n$, then the conjectured relation becomes
\begin{align*}
&\alpha_{2,1,1}\, x_{\val_{\frac32}(u)} + \cdots +\alpha_{2,1,6}\, x_{\val_{\frac32}(u022)}  = x_{\val_{\frac32}(u011)}  \\
&\Leftrightarrow
5(\underbrace{\val_{\frac32}(u01)}_{=18n+5})^2+5(\underbrace{\val_{\frac32}(u20)}_{=18n+6})^2-(\underbrace{\val_{\frac32}(u22)}_{=18n+7})^2=4(\underbrace{\val_{\frac32}(u011)}_{=27n+8})^2
\end{align*}
and holds since $(\val_{\frac32}(u011))^2= 729 n^2+432 n+64$.
      \end{example}

      \begin{theorem}
        The sequence $(n^2)_{n\ge 0}$ is $(\mathbb{Q},\frac32)$-regular.
      \end{theorem}

      \begin{proof}
        The result follows from the next $27$ relations that can be easily deduced from Equation~\eqref{eq:numsys}.
        We start with the relations found in the above example. Note that for two distinct remainders, the corresponding trees have disjoint domains. If $\val_{\frac32}(u)\equiv 2\pmod{8}$,
        $$
        \begin{array}{rcl}
            5(\val_{\frac32}(u01))^2+5(\val_{\frac32}(u20))^2-(\val_{\frac32}(u22))^2&=&4(\val_{\frac32}(u011))^2  \\    9(\val_{\frac32}(u20))^2&=&4(\val_{\frac32}(u200))^2\\
        -(\val_{\frac32}(u01))^2+5(\val_{\frac32}(u20))^2+5(\val_{\frac32}(u22))^2&=&4(\val_{\frac32}(u202))^2\\
        2(\val_{\frac32}(u01))^2-7(\val_{\frac32}(u20))^2+14(\val_{\frac32}(u22))^2&=&4(\val_{\frac32}(u202))^2.\\
        \end{array}$$
        If $\val_{\frac32}(u)\equiv 3\pmod{8}$,
        $$
        \begin{array}{rcl}
         9(\val_{\frac32}(u11))^2&=&4(\val_{\frac32}(u110))^2\\
         189(\val_{\frac32}(u))^2-345(\val_{\frac32}(u1))^2+161(\val_{\frac32}(u1))^2&=&20(\val_{\frac32}(u112))^2.\\
        \end{array}$$
        If $\val_{\frac32}(u)\equiv 4\pmod{8}$,
        $$
        \begin{array}{rcl}
          5(\val_{\frac32}(u00))^2+5(\val_{\frac32}(u02))^2-(\val_{\frac32}(u21))^2&=&4(\val_{\frac32}(u001))^2\\
          9(\val_{\frac32}(u02))^2&=&4(\val_{\frac32}(u020))^2\\
          -(\val_{\frac32}(u00))^2+5(\val_{\frac32}(u02))^2+5(\val_{\frac32}(u21))^2&=&4(\val_{\frac32}(u022))^2\\
            2(\val_{\frac32}(u00))^2-7(\val_{\frac32}(u02))^2+14(\val_{\frac32}(u21))^2&=&4(\val_{\frac32}(u211))^2.\\
        \end{array}$$
         If $\val_{\frac32}(u)\equiv 5\pmod{8}$,
        $$
        \begin{array}{rcl}
          9(\val_{\frac32}(u10))^2&=&4(\val_{\frac32}(u100))^2\\
          -162(\val_{\frac32}(u))^2+119(\val_{\frac32}(u10))^2+102(\val_{\frac32}(u12))^2&=&84(\val_{\frac32}(u102))^2\\
                   324(\val_{\frac32}(u))^2-175(\val_{\frac32}(u10))^2+300(\val_{\frac32}(u12))^2&=&84(\val_{\frac32}(u121))^2.\\
        \end{array}$$
         If $\val_{\frac32}(u)\equiv 6\pmod{8}$,
        $$
        \begin{array}{rcl}
          9 (\val_{\frac32}(u01))^2  &=& 4 (\val_{\frac32}(u010))^2 \\
          2 (\val_{\frac32}(u01))^2 + 8 (\val_{\frac32}(u20))^2 - (\val_{\frac32}(u22))^2 &=& 4 (\val_{\frac32}(u012))^2 \\
          - (\val_{\frac32}(u01))^2 + 8 (\val_{\frac32}(u20))^2 +2 (\val_{\frac32}(u22))^2 &=& 4 (\val_{\frac32}(u201))^2 \\
          9 (\val_{\frac32}(u22))^2 &=& 4 (\val_{\frac32}(u220))^2 \\
          5 (\val_{\frac32}(u01))^2 - 16 (\val_{\frac32}(u20))^2 +20 (\val_{\frac32}(u22))^2 &=& 4 (\val_{\frac32}(u222))^2.\\
        \end{array}$$
        If $\val_{\frac32}(u)\equiv 7\pmod{8}$,
        \begin{equation}
          \label{eq:unique}
          27 (\val_{\frac32}(u))^2 - 57 (\val_{\frac32}(u1))^2 + 38 (\val_{\frac32}(u11))^2 = 8 (\val_{\frac32}(u111))^2.
        \end{equation}
        Before moving to the last two remainders, let us point out that we have here a single equation because the corresponding tree has a linear form (each node has a single child). Its domain is restricted to $\{\varepsilon,1,11,111\}$. This is the only relation where we need the decorations of the nodes on every level.
        If $\val_{\frac32}(u)\equiv 0\pmod{8}$,
         $$
        \begin{array}{rcl}
          9 (\val_{\frac32}(u00))^2  &=&  4 (\val_{\frac32}(u000))^2\\
          2 (\val_{\frac32}(u00))^2 + 8  (\val_{\frac32}(u02))^2 - (\val_{\frac32}(u21))^2 &=&  4 (\val_{\frac32}(u002))^2\\
          - (\val_{\frac32}(u00))^2 + 8 (\val_{\frac32}(u02))^2  +2 (\val_{\frac32}(u21))^2 &=& 4 (\val_{\frac32}(u021))^2\\
          9 (\val_{\frac32}(u21))^2 &=&  4(\val_{\frac32}(u210))^2\\
          5 (\val_{\frac32}(u00))^2 -16 (\val_{\frac32}(u02))^2 +20  (\val_{\frac32}(u21))^2 &=& 4 (\val_{\frac32}(u212))^2.\\
        \end{array}$$
        If $\val_{\frac32}(u)\equiv 1\pmod{8}$,
         $$
        \begin{array}{rcl}
          -162 (\val_{\frac32}(u))^2 +182  (\val_{\frac32}(u10))^2 + 39 (\val_{\frac32}(u12))^2 &=& 84 (\val_{\frac32}(u101))^2\\
          9 (\val_{\frac32}(u12))^2 &=& 4 (\val_{\frac32}(u120))^2\\
          810 (\val_{\frac32}(u))^2 - 406 (\val_{\frac32}(u10))^2 + 435 (\val_{\frac32}(u12))^2 &=& 84 (\val_{\frac32}(u122))^2.\\
        \end{array}$$
We conclude that the sequence of squares is indeed $(\mathbb{Q},\frac32)$-regular.
      \end{proof}

If the considered numeration system has the following additional property, then we have a more elegant result.
The base-$\frac{p}{q}$ numeration system is {\em expanding} if for all $\frac{p}{q}$-representations $w$ and for all $h\ge 1$, the number of leaves of $T[w,h+1]$ is greater than the number of leaves of $T[w,h]$.
As an example, base~$\frac32$ is not expanding because the numeration tree contains a factor whose domain is $\{1^n\mid 0\le n\le h\}$ for all $h$. This was exactly the case $\val_{\frac32}(u)\equiv 7\pmod{8}$ in the above proof. Such a particularity will be useful in the next subsection, see Example~\ref{exa:ceaut}.  To the contrary, base~$\frac52$ has the expanding property because its signature is $(024, 13)^\omega$ and thus every node in the numeration tree has degree at least $2$.

\begin{theorem}
Suppose that the base-$\frac{p}{q}$ numeration system is expanding. For all integer $d\ge 1$, the sequence $(n^d)_{n\ge 0}$ is $(\mathbb{Q},\frac{p}{q})$-regular. 
\end{theorem}

\begin{proof}
  By assumption, if we consider factors of increasing height, then they have more and more leaves: there exists an integer $h>1$ such that for all representations $w$ in base $\frac{p}{q}$, the tree factor $T[w,h-1]$ has at least $d+1$ leaves. Consider a tree of the form $T[w,h]$ where $\val_{\frac{p}{q}}(w)=q^h n+r$ for some $0\le r<q^h$. 
The lowest level of internal nodes contains at least $d+1$ nodes that have consecutive $\frac{p}{q}$-numerical values (because of the breadth-first ordering). Assume that the first node on this level is of the form $wx$ with $|x|=h-1$. Hence, by Equation~\eqref{eq:numsys}, its value is given by
  $$\val_{\frac{p}{q}}(wx)=\val_{\frac{p}{q}}(w)\left(\frac{p}{q}\right)^{|x|}+\val_{\frac{p}{q}}(x)=
  (q^h n+r)\left(\frac{p}{q}\right)^{h-1}+\val_{\frac{p}{q}}(x),$$
  which we denote by $\alpha n+\beta$ with $\alpha=q\, p^{h-1}$ and $\beta\in\mathbb{Q}$. Now consider a leaf of the form $wy$ with $|y|=h$. Its value is given by 
  $$\val_{\frac{p}{q}}(wy)=
  (q^h n+r)\left(\frac{p}{q}\right)^{h}+\val_{\frac{p}{q}}(y),$$
  which we denote by $\mu n+\nu$ with $\mu=p^{h}$ and $\nu\in\mathbb{Q}$. We claim that the decoration $(\mu n+\nu)^d$ of this leaf can be expressed as a linear combination of the $d+1$ decorations $$(\alpha n+\beta)^d,(\alpha n+\beta+1)^d,\ldots, (\alpha n+\beta+d)^d$$ of the leaves on the above level because these polynomials form a basis of the $\mathbb{Q}$-vector space of the polynomials in $\mathbb{Q}[x]$ of degree less than or equal to $d$. Indeed, develop these polynomials using the binomial theorem and write down the coefficients of $n^j$ as rows in the following matrix:
  $$
  \begin{pmatrix}
    \alpha^d & \binom{d}{1} \alpha^{d-1} \beta & \binom{d}{2} \alpha^{d-2} \beta^2 & \cdots & \binom{d}{d-1} \alpha \beta^{d-1}& \beta^d \\
    \alpha^d & \binom{d}{1} \alpha^{d-1} (\beta+1) & \binom{d}{2} \alpha^{d-2} (\beta+1)^2 & \cdots & \binom{d}{d-1} \alpha (\beta+1)^{d-1}& (\beta+1)^d \\
    \vdots & & & & & \vdots \\
      \alpha^d & \binom{d}{1} \alpha^{d-1} (\beta+d) & \binom{d}{2} \alpha^{d-2} (\beta+d)^2 & \cdots & \binom{d}{d-1} \alpha (\beta+d)^{d-1}& (\beta+d)^d \\  
  \end{pmatrix}.$$
  Its determinant is given by
  $$  \alpha^{d!}\binom{d}{1}\cdots \binom{d}{d-1}\det
  \begin{pmatrix}
    1 & \beta & \cdots &\beta^d \\
    1 & \beta+1 & \cdots & (\beta+1)^d\\
    \vdots & & & \vdots \\
    1 & \beta+d & \cdots & (\beta+d)^d\\
  \end{pmatrix},
  $$
  which is non-zero because we have a Vandermonde matrix.
  We have just proven that the decorated tree $T_{(n^d)_{n\ge 0}}(L_{\frac{p}{q}})$ is $h$-linear, so the sequence $(n^d)_{n\ge 0}$ is $(\mathbb{Q},\frac{p}{q})$-regular.
\end{proof}

\begin{corollary}
Suppose that the base-$\frac{p}{q}$ numeration system is expanding. Let $P\in\mathbb{Q}[x]$ be a non-constant polynomial. The sequence $(P(n))_{n\ge 0}$ is $(\mathbb{Q},\frac{p}{q})$-regular. 
\end{corollary}

\begin{proof}
  This is a direct consequence of the above theorem and Proposition~\ref{pro:module}.
\end{proof}

Based on computer experiments, we conjecture that the above results also hold for rational base systems without the expanding property, such as for base~$\frac32$. In that case, the regularity of the sequence $(n^d)_{n\ge 0}$ should be achieved by taking factors of height at least $d+1$. One expects to express the decoration of the unique leaf of a linear tree, whose domain contains a single word of each length, from the $d+1$ decorations of the nodes above as for relation \eqref{eq:unique}.

\subsection{Sequences taking finitely many values}
In this section, we study sequences taking finitely many values.
With Proposition~\ref{pro: p/q-reg implies p/q-aut}, we show that a $(\mathbb{L},\frac{p}{q})$-regular sequence taking finitely many values is always $\frac{p}{q}$-automatic. We prove that, however, the converse does not hold in general.

\begin{definition}
Let $B\subseteq \mathbb{K}$ be a finite alphabet.
 A sequence $\mathbf{x}=x_0x_1\cdots\in B^\mathbb{N}$ is  {\em $\frac{p}{q}$-automatic} if there exists a deterministic finite automaton with output (DFAO for short) $\mathcal{A}=(Q,q_0,A_p,\delta,\mu:Q\to B)$ such that $x_n=\mu(\delta(q_0,\rep_{\frac{p}{q}}(n)))$ for all $n\ge 0$.
\end{definition}

We recall some notation and results from \cite{RS}.

\begin{definition}
For all $h\ge 0$, we let $F_h$ denote the set of factors of height $h$ occurring in $T_\mathbf{x}(L_{\frac{p}{q}})$, and we let $F_h^\infty \subseteq F_h$ denote the subset of elements in $F_h$ occurring infinitely often in $T_\mathbf{x}(L_{\frac{p}{q}})$. 
For any letter $a\in A_p$, we let $F_{h,a}^\infty \subseteq F_h^\infty$ denote the set of factors of height $h$ occurring infinitely often in $T_\mathbf{x}(L_{\frac{p}{q}})$ such that the label of the edge between the first node on level $h-1$ and its first child is $a$. Otherwise stated, the first word of length $h$ by lexicographic order in the domain of the factor ends with~$a$.
\end{definition}

In the signature $(w_0,\ldots,w_{q-1})$ of $T_\mathbf{x}(L_{\frac{p}{q}})$, we let $w_{j,0}$ denote the first symbol of $w_j$, for all $0 \le j\le q-1$.

\begin{theorem}\cite[Theorem~38]{RS}\label{thm: RS-Thm 36}
  Let $\mathbf{x}$ be a sequence over a finite alphabet~$B$.
  If there exists some $h\ge 0$ such that
  $\# F_{h+1,w_{j,0}}^\infty\le \# F_{h}^\infty$ for all $0 \le j\le q-1$, then $\mathbf{x}$ is $\frac{p}{q}$-automatic.
\end{theorem}

\begin{prop}\label{pro: p/q-reg implies p/q-aut}
If a sequence $\mathbf{x}$ is $(\mathbb{L},\frac{p}{q})$-regular and takes only finitely many values, then it is $\frac{p}{q}$-automatic.
\end{prop}
\begin{proof}
Suppose that $\mathbf{x}$ is $(\mathbb{L},\frac{p}{q})$-regular and takes its values in the finite alphabet $B\subseteq \mathbb{K}$.
Then the decorated tree $T_\mathbf{x}(L_{\frac{p}{q}})$ is $(\mathbb{L},h)$-linear for some $h\ge 1$.
This means that the decorations on level $h$ of a given element of $F_{h}^\infty$ are uniquely determined by the decorations of its first $h-1$ levels.
In particular, $\# F_{h,w_{j,0}}^\infty\le \# F_{h-1}^\infty$ for all $0 \le j\le q-1$.
We now conclude by Theorem~\ref{thm: RS-Thm 36}.
\end{proof}

\begin{prop}
Let $m\ge 2$.  If a sequence $\mathbf{x}\in\mathbb{Z}^{\mathbb{N}}$ is $(\mathbb{Z},\frac{p}{q})$-regular, then the sequence $\mathbf{y}=(x_n \bmod{m})_{n\ge 0}$ is $\frac{p}{q}$-automatic.
\end{prop}

\begin{proof}
  The tree $T_\mathbf{x}(L_\frac{p}{q})$ is $h$-linear for some $h\ge 1$ and thus it is also $(h+1)$-linear by Lemma~\ref{lem:h+1}.  Since the decorations of $T_\mathbf{y}(L_\frac{p}{q})$ are in $\{0,\ldots,m-1\}$, there is a finite number of decorated factors of height~$h$ in $T_\mathbf{y}(L_\frac{p}{q})$. Take a factor in $F_{h}^\infty$ and consider its possible extensions to a factor in $F_{h+1,w_{j,0}}^\infty$. Such an extension is unique because taking modulo $m$ any linear relation occurring in the $(h+1)$-linear tree $T_\mathbf{x}(L_\frac{p}{q})$ implies that the decorations on the last level of any factor of height $h+1$ in $T_\mathbf{y}(L_\frac{p}{q})$ depends on the decorations above. We can therefore apply Theorem~\ref{thm: RS-Thm 36}.
\end{proof}

Due to the non-regularity of the language $L_{\frac{p}{q}}$, we cannot expect a converse of the above proposition, i.e., there are $\frac{p}{q}$-automatic sequences which are not $\frac{p}{q}$-regular even in the case where the states of the DFAO have pairwise distinct outputs (which is supposedly easier to handle since there is no ambiguity about the current state when an output is produced, see the discussion in \cite[Remark~37]{RS}).

Let us consider the following example in base $\frac{3}{2}$ of an automatic sequence that is not regular.

\begin{example}\label{exa:ceaut} We will make use of the fact that the base-$\frac{3}{2}$ numeration is non-expanding. 
Consider the sequence $\mathbf{x}=q_0q_1q_2131551333223553\cdots$ generated by the DFAO in Figure~\ref{fig:counterex} in base $\frac{3}{2}$.
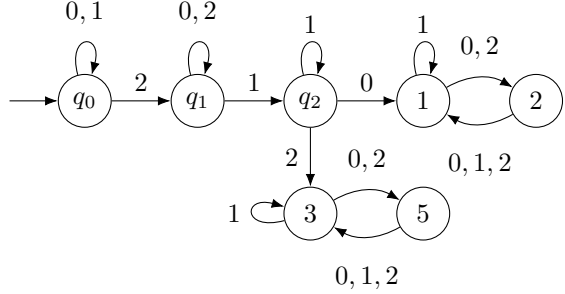
\begin{figure}[htb]
\begin{center}
\begin{tikzpicture}
  \tikzstyle{every node}=[shape=circle, fill=none, draw=black,minimum size=20pt, inner sep=2pt]
\node(1) at (0,0) {$q_0$};
\node(2) at (1.5,0) {$q_1$};
\node(3) at (3,0) {$q_2$};
\node(4) at (4.5,0) {$1$};
\node(5) at (6,0) {$2$};
\node(6) at (3,-1.5) {$3$};
\node(7) at (4.5,-1.5) {$5$};

\tikzstyle{every node}=[shape=circle, minimum size=5pt, inner sep=2pt]

\draw [-Latex] (-1,0) to node [above] {} (1);

\draw [-Latex] (1) to [] node [above] {$2$} (2);
\draw [-Latex] (2) to [] node [above] {$1$} (3);
\draw [-Latex] (3) to [] node [above] {$0$} (4);
\draw [-Latex] (3) to [] node [left] {$2$} (6);
\draw [-Latex] (4) to [loop above] node [above] {$1$} (4);
\draw [-Latex] (4) to [bend left] node [above] {$0,2$} (5);
\draw [-Latex] (5) to [bend left] node [below] {$0,1,2$} (4);
\draw [-Latex] (6) to [loop left] node [left] {$1$} (6);
\draw [-Latex] (6) to [bend left] node [above] {$0,2$} (7);
\draw [-Latex] (7) to [bend left] node [below] {$0,1,2$} (6);
\draw [-Latex] (1) to [loop above] node [above] {$0,1$} (1);
\draw [-Latex] (2) to [loop above] node [above] {$0,2$} (2);
\draw [-Latex] (3) to [loop above] node [above] {$1$} (3);
\end{tikzpicture}
\caption{A DFAO with 7 states reading $\frac{3}{2}$-representations.}
\label{fig:counterex}
\end{center}
\end{figure}
We have chosen outputs $1,2,3,5$ because the vectors $(1,2)$ and $(3,5)$ are linearly independent.
The outputs $q_0$, $q_1$ and $q_2$ are not relevant.

We show that the sequence $\mathbf{x}$ is not $\frac{3}{2}$-regular. We proceed by contradiction. Let $h\ge 1$ and assume that $T_\mathbf{x}(L_\frac{3}{2})$ is $h$-linear. Consider the word $s=1^{h-1}0\in A_3^*$. We make use of basic properties of rational base numeration systems. There exists a word $w_0$ such that $w_0s\in L_{\frac{3}{2}}$ (see, for instance, \cite[Corollaire~4.17]{Marsault-these}) and, for all $k\ge 0$, the word $w_k$ representing the integer $\val_{\frac32}(w_0)+k\cdot 2^h$ is such that $w_ks\in L_{\frac{3}{2}}$ (see, for instance, \cite[Lemme~4.14]{Marsault-these}). Analyzing the numeration tree (where the number of nodes on each level is increasing), the functions
$$g_0:n\mapsto \# \{ x\in\{0,1,2\}^n \mid 210x \in L_{\frac{3}{2}}\}$$
and
$$g_2:n\mapsto \# \{ x\in\{0,1,2\}^n \mid 212x \in L_{\frac{3}{2}}\}$$
are both increasing for $n\ge 2$. Note that $g_0(n)+g_2(n)$ is the number of words of length $n+3$ in $L_\frac{3}{2}$. So there exists a large enough integer $N$ such that $g_0(N)>2^h$ and $g_2(N)>2^h$. Without loss of generality, we may assume that $N$ is even. So there exist words $w_i,w_j \in A_3^*$ respectively of the form $210x$ and $212x'$ such that $w_is,w_js\in L_{\frac{3}{2}}$ and $|w_i|=|w_j|$ is odd.

\begin{figure}[h!tb]
  \centering
    \tikzset{
  s_gra/.style = {circle,fill=gray!30, inner sep=1pt, minimum size=15pt}
}
\tikzstyle{level 1}=[sibling distance=60mm]
\tikzstyle{level 2}=[sibling distance=50mm]
\tikzstyle{level 3}=[sibling distance=40mm]
\tikzstyle{level 4}=[sibling distance=15mm]
\tikzstyle{level 5}=[sibling distance=5mm]

\begin{tikzpicture}[->,>=stealth',level distance = 1cm]
  \node [s_gra] {$1$} 
   child {node [s_gra] {$1$} 
        child {node [s_gra] {$1$}
          child {node [s_gra] {$1$}
            child {node [s_gra] {$2$} edge from parent node[above] {$0$}}
            child {node [s_gra] {$2$} edge from parent node[above] {$2$}}
            edge from parent node[left] {$1$}}
          edge from parent node[left] {$1$}} 
            edge from parent node[left] {$1$}}
;
\end{tikzpicture}
\quad 
\begin{tikzpicture}[->,>=stealth',level distance = 1cm]
  \node [s_gra] {$3$} 
   child {node [s_gra] {$3$} 
        child {node [s_gra] {$3$}
          child {node [s_gra] {$3$}
            child {node [s_gra] {$5$} edge from parent node[above] {$0$}}
            child {node [s_gra] {$5$} edge from parent node[above] {$2$}}
            edge from parent node[left] {$1$}}
          edge from parent node[left] {$1$}} 
            edge from parent node[left] {$1$}}
;
\end{tikzpicture}
\caption{The trees $T[w_i,h]$ and $T[w_j,h]$ for $h=4$.}\label{fig:counterex2}
\end{figure}
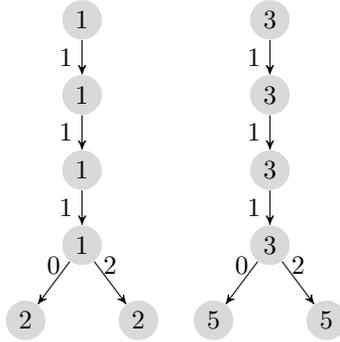

In the decorated tree $T_\mathbf{x}(L_\frac{3}{2})$, the factor $T[w_i,h]$ has all its decorations equal to $1$, except on the last level, where the decorations of the two nodes are equal to $2$ (the decoration of the root is also $1$ since $|w_i|$ is odd). See Figure~\ref{fig:counterex2} for an illustration when $h=4$. Such a factor has $h+2$ nodes, one on every internal level and two on the last level. As labeled trees, the factors $T[w_i,h]$ and $T[w_j,h]$ are equal (they have the same structure). Nevertheless, the factor $T[w_j,h]$ has all its decorations equal to $3$, except on the last level, where the decorations are equal to $5$. See Figure~\ref{fig:counterex2} again.  
Since we assume $h$-linearity, a linear relation must exist for the first node on the last level, linking its decoration with the decorations of the nodes above, and in both trees, the relation is the same.
More precisely, there must exist coefficients $c_0, c_1, \ldots, c_{h-1}$ such that
\begin{align*}
2 &= c_0 \cdot 1 + \cdots +  c_{h-1} \cdot 1, \\
5 &= c_0 \cdot 3 + \cdots +  c_{h-1} \cdot 3
\end{align*}
both hold. This is impossible since in the second tree all decorations are equal to three times the ones of the first tree, the same multiplicative relation should occur on the last level but $2\cdot 3\neq 5$. The same reasoning can be done for an arbitrary $h$.
\end{example}

\subsection{Cumulative version of a sequence}

In this section, we define and study the cumulative version of a sequence.
For any abstract numeration system $S=(L,A,<)$, recall that integers are in one-to-one correspondence with words in $L$, so we may index a sequence by words in $L$, i.e., we write $x_w$ instead of $x_n$ whenever $\rep_S(n)=w$. 
\begin{definition}
Suppose that the numeration language $L$ is prefix-closed.
  For any sequence $\mathbf{x}\in \mathbb{K}^\mathbb{N}$, its \emph{cumulative version} is the sequence $\mathbf{y}\in \mathbb{K}^\mathbb{N}$ defined, for all words $w\in L$, by
\[
y_w = \sum_{u \in \mathrm{Pref}(w)} x_u,
\]
where the sum goes over all prefixes $\varepsilon,w_1,w_1w_2,\ldots,w_1\cdots w_k$ of $w=w_1\cdots w_k$.
\end{definition}

\begin{example}
We consider the base-$\frac{3}{2}$ case. The sequence $\mathbf{x}=(n)_{n\ge 0}$ and its cumulative version are depicted in Figure~\ref{fig:ex-cumulative}.
For instance, from the root of the tree $T( L_{\frac{3}{2}})$, the third node on the last level is reached by reading the word $2120$, so its decoration in the second tree is given by \[
x_\varepsilon + x_{2} + x_{21} + x_{212} + x_{2120} 
=0+1+2+4+6
= 13.
\]
 \begin{figure}[h!t]
    \centering
    \tikzset{
  s_gra/.style = {circle,fill=gray!30, inner sep=1pt, minimum size=15pt}
}
\tikzstyle{level 1}=[sibling distance=40mm]
\tikzstyle{level 2}=[sibling distance=30mm]
\tikzstyle{level 3}=[sibling distance=20mm]
\tikzstyle{level 4}=[sibling distance=10mm]
\begin{tikzpicture}[->,>=stealth',level distance = 1cm]
  \node [s_gra] {$0$} 
  child { node [s_gra] {$1$} 
    child { node [s_gra] {$2$} 
      child {node [s_gra] {$3$} 
        child {node [s_gra] {$5$}
          edge from parent node[left] {$1$}} 
            edge from parent node[above] {$0$}}
    child {node [s_gra] {$4$} 
      child {node [s_gra] {$6$}
        edge from parent node[left] {$0$}} 
      child {node [s_gra] {$7$}
        edge from parent node[right] {$2$}} 
  edge from parent node[above] {$2$}}
edge from parent node[left] {$1$}}
edge from parent node[left] {$2$}}  
;
\end{tikzpicture}
\quad 
\begin{tikzpicture}[->,>=stealth',level distance = 1cm]
  \node [s_gra] {$0$} 
  child { node [s_gra] {$1$} 
    child { node [s_gra] {$3$} 
      child {node [s_gra] {$6$} 
        child {node [s_gra] {$11$}
          edge from parent node[left] {$1$}} 
            edge from parent node[above] {$0$}}
    child {node [s_gra] {$7$} 
      child {node [s_gra] {$13$}
        edge from parent node[left] {$0$}} 
      child {node [s_gra] {$14$}
        edge from parent node[right] {$2$}} 
  edge from parent node[above] {$2$}}
edge from parent node[left] {$1$}}
edge from parent node[left] {$2$}}  
;
\end{tikzpicture}
    \caption{The first few levels of the trees in base $\frac{3}{2}$ respectively decorated by the sequence $(n)_{n\ge 0}$ and its cumulative version.}
    \label{fig:ex-cumulative}
  \end{figure}
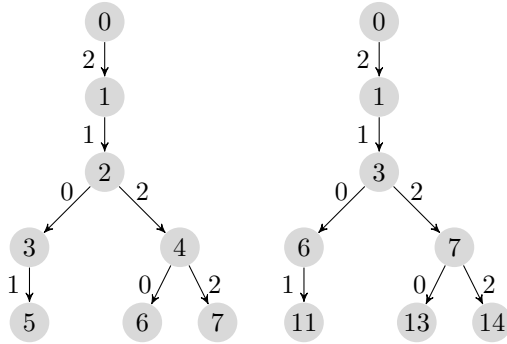

  For the sum-of-digits sequence $\mathbf{s}$, this cumulative version is a sum weighted by the position of the digit, i.e., for a word $w_k\cdots w_0$, we get $k.w_k+\cdots +1.w_1$.
\end{example}

\begin{theorem}
Let $S=(L,A,<)$ be an abstract numeration system for which the language $L$ is prefix-closed. 
  Let $\mathbf{x}$ be a sequence such that  the decorated tree $T_\mathbf{x}(L)$ is $(\mathbb{L},h)$-linear for some $h\ge 1$. The tree decorated by the cumulative version of $\mathbf{x}$ is $(\mathbb{L},h+1)$-linear.
\end{theorem}
\begin{proof}
Let $\mathbf{y}$ denote the cumulative version of $\mathbf{x}$. Let $w\in L$. Take a word $u\in A^{h-1}$ and two letters $a,b\in A$ such that $aub \in \dom(T[w,h+1])$. Note that $|aub|=h+1\ge 2$. By definition of the sequence $\mathbf{y}$, we have $y_{waub} = y_{wau} + x_{waub}$.
By assumption on the sequence $\mathbf{x}$, there exist constants $c_{wa,ub,v}\in\mathbb{L}$ such that 
\[
x_{waub}=\sum_{\substack{v\in \dom(T[wa,h])\\ |v|<h}} c_{wa,ub,v}\, x_{wav}.
\]
If $z=z_\ell\cdots z_1z_0$ is a non-empty word, we let $\text{del}(z)=z_\ell\cdots z_1$ be the word obtained by deleting the last letter.
By definition of the sequence $\mathbf{y}$,  we have $x_{wav} = y_{wav} - y_{\text{del}(wav)}$. 
Therefore 
\[
y_{waub} = y_{wau} + \sum_{\substack{v\in \dom(T[wa,h])\\ |v|<h}} c_{wa,ub,v}\, (y_{wav} - y_{\text{del}(wav)}),
\]
which means that the decoration $y_{waub}$ of a leaf of $T[w,h+1]$ can be expressed as a linear combination of the decorations of the nodes of $T[w,h]$ (in the above sum, for $|v|=0$, $\text{del}(wav)=w$ is the root of $T[w,h]$ and for $|v|=h-1$, $|av|=|au|=h$).
Therefore $T_\mathbf{y}(L)$ is $(\mathbb{L},h+1)$-linear.
\end{proof}

As a consequence of Theorem~\ref{the:kautom}, we have the following result.

\begin{corollary}
Let $\mathbf{x}\in\mathbb{K}^\mathbb{N}$ be a $(\mathbb{L},k)$-regular sequence.
Then its cumulative version is also $(\mathbb{L},k)$-regular.
\end{corollary}

As a consequence of Theorem~\ref{the:Sreg}, we have the following result.

\begin{corollary}\label{cor:cumulative for ANS reg and prefixed closed}
  Let $S=(L,A,<)$ be an abstract numeration system built on a prefix-closed regular language. If a sequence $\mathbf{x}\in\mathbb{K}^\mathbb{N}$ is $(\mathbb{L},S)$-regular, then its cumulative version is also $(\mathbb{L},S)$-regular.
\end{corollary}

Finally, from Definition~\ref{def:Sregular}, we also have the following.

\begin{corollary}\label{cor:cumulative for rational bases}
Let $\mathbf{x}$ be a $(\mathbb{L},\frac{p}{q})$-regular sequence.
Then its cumulative version is also $(\mathbb{L},\frac{p}{q})$-regular.
\end{corollary}

For all words $u,w$, we let $|u|_w$ denote the number of occurrences of the factor $w$ in $u$.

\begin{corollary}\label{cor:Zs}
Let $S=(L,A,<)$ be an abstract numeration system built on a prefix-closed regular language. For any word $w\in A^*$, the sequence $(|\rep_S(n)|_w)_{n\ge 0}$ is $(\mathbb{Z},S)$-regular.
\end{corollary}
\begin{proof}
Fix a word $w\in A^*$.
Consider a DFAO over $A$ whose output is $1$ if and only if the input word has the suffix $w$. Otherwise, the output is $0$. This DFAO produces an $S$-automatic sequence $\mathbf{x}$.
By~\cite[Theorem 38]{CCS}, the sequence $\mathbf{x}$ is in particular $(\mathbb{Z},S)$-regular.
The cumulative version of $\mathbf{x}$ is exactly the sequence $(|\rep_S(n)|_w)_{n\ge 0}$.
Therefore by Corollary~\ref{cor:cumulative for ANS reg and prefixed closed}, the latter sequence is also $(\mathbb{Z},S)$-regular.
\end{proof}

\begin{corollary}
Let $w\in A_p^*$ be a word. The sequence $(|\rep_{\frac{p}{q}}(n)|_w)_{n\ge 0}$ is $(\mathbb{Q},\frac{p}{q})$-regular.
\end{corollary}
Because of Example~\ref{exa:ceaut}, we cannot proceed as in the proof of Corollary~\ref{cor:Zs}: in this setting, a $\frac{p}{q}$-automatic sequence is not always $\frac{p}{q}$-regular. 
\begin{proof}
We start by defining an auxiliary sequence through decorations in the tree $T(L_\frac{p}{q})$ as follows. The only two decorations are $1$ and $2$ (we have chosen these numbers instead of $0,1$ because they are invertible). 
  If $u\in L_\frac{p}{q}$ ends with $w$, then the decoration of the node $u$ is $2$, otherwise it is $1$. In particular, this defines a sequence $\mathbf{x}=(x_n)_{n\ge 0}$ such that $x_n$ is the decoration associated with the $\frac{p}{q}$-representation of $n$. We will show that the sequence $\mathbf{x}$ is $(\mathbb{Q},\frac{p}{q})$-regular by showing that the corresponding decorated tree $T_\mathbf{x}(L_\frac{p}{q})$ is $(|w|+1)$-linear.

  Let $T[y,|w|+1]$ be a factor of $T_\mathbf{x}(L_\frac{p}{q})$. The decorations of the nodes on the $|w|-1$ first levels depend on the ``past'' of the factor: one has to know the last $|w|$ letters of the path from the root of $T(L_\frac{p}{q})$ to such a node to determine the decorations. This cannot be done inside the window given by $T[y,|w|+1]$. Nevertheless, the nodes on the last two levels of $T[y,|w|+1]$ are completely determined within the factor. In particular, the decoration of a node on the last level can be expressed from the decoration of a node on the previous level.

  Assume that $yw\in L_\frac{p}{q}$, i.e., $w$ belongs to the domain of $T[y,|w|+1]$. Hence $x_{yw}=2$. Now consider nodes on the last level. If there exists a digit $d$ such that $ydw$ belongs to the domain of $T[y,|w|+1]$, then $x_{ydw}=x_{yw}$. For all the other words $v$ of length $|w|+1$ such that $yv$ belongs to the domain of $T[y,|w|+1]$, then $x_{yv}=\frac{x_{yw}}{2}$. If there is a word $z\neq w$ of length $|w|$ such that $yz\in L_\frac{p}{q}$, then $x_{yz}=1$. We can therefore replace the above two equations by $x_{ydw}=2x_{yz}$ and $x_{yv}=x_{yz}$.

  So the sequence $\mathbf{x}$ indicating all $\frac{p}{q}$-representation ending with $w$ is $(\mathbb{Q},\frac{p}{q})$-regular.
 By Proposition~\ref{pro:module}, the sequence $\mathbf{x}-(1)_{n\ge 0}$ is also $(\mathbb{Q},\frac{p}{q})$-regular.
The cumulative version of the latter sequence is exactly the sequence $(|\rep_\frac{p}{q}(n)|_w)_{n\ge 0}$.
We conclude by using Corollary~\ref{cor:cumulative for rational bases}.
\end{proof}




          
\section{Graph directed linear representations}\label{sec: matrices}

In this section, we consider a prefix-closed language $L\subseteq A^*$, and when necessary, the abstract numeration system $S=(L,A,<)$ built on $L$. 
Let $T_\mathbf{x}(L)$ be a $(\mathbb{L},h)$-linear decorated tree for some $h\ge 1$. Our aim is to define the analogue of a linear representation of a rational series.

Consider the restriction $D$ to words of length less than $h$ of the union of the domains of the pairwise distinct factors $T_1,\ldots,T_{r(h)}$ of height~$h$ occurring in the undecorated tree $T(L)$, i.e.,
$$D=\bigcup_{j=1}^{r(h)} \dom(T_j) \cap A^{<h}.$$ We enumerate the words of $D$ by radix order. We will consider matrices in $\mathbb{L}^{D\times D}$ and column vectors in $\mathbb{L}^D$. A vector in $\mathbb{L}^D$ encodes the decorations of the nodes of the prefix of height $h-1$ of any factor of height $h$ occurring in $T_\mathbf{x}(L)$. Some entry may be equal to $0$ whenever the corresponding word of $D$ does not belong to the domain of the considered factor.
\begin{example}[Coding vectors]
In Example~\ref{exa:fib}, $h=2$ and $D=\{\varepsilon<0<1\}$. So we consider column vectors of dimension $3$ encoding the corresponding decorations $x_{uz}$ of the nodes $uz$, $z\in D$, of the considered factors $T[u,1]$ of height $h-1=1$. Hence, $T[100,1]$, $T[1000,1]$ and $T[101,1]$ are respectively coded (see Figure~\ref{fig:ex2b}) by
$$
\begin{pmatrix}
  x_{u}\\
  x_{u0}\\
  x_{u1}\\
\end{pmatrix}=
\begin{pmatrix}
  4\\5\\6\\
\end{pmatrix},\
\begin{pmatrix}
  5\\6\\8\\
\end{pmatrix},\
\begin{pmatrix}
  4\\6\\0\\
\end{pmatrix}.$$
Since $1$ does not belong to domain of $T[101,1]$, the last entry is set to $0$.
\end{example}

A word $w=w_0w_1\cdots w_\ell$ in $L$, with $\ell\ge h-1$, defines a path in $T_\mathbf{x}(L)$ and thus a sequence of factors of height $h$ of $T_\mathbf{x}(L)$ (given by the sequence of respective roots):
\begin{equation}
  \label{eq:seqtree}
  T[\varepsilon,h],\ T[w_0,h],\ T[w_0w_1,h],\ \ldots,\ T[w_0\cdots w_{\ell-h},h].
\end{equation}
In particular, it also provides us with a sequence of factors of height $h-1$ of $T_\mathbf{x}(L)$ and thus a sequence of column vectors in $\mathbb{L}^D$. Since the tree is $h$-linear, we will explain that this sequence of vectors can be computed by convenient successive matrix multiplications until we get the desired decoration.

In the following lemma, we assume that the types of the undecorated trees occurring in the sequence \eqref{eq:seqtree} are known. This assumption has to be checked for the numeration system of interest --- this will be discussed later on.

\begin{lemma}\label{lem:matrices}
Let $0\le i<\ell-h$. Let $z_i,z_{i+1}\in\mathbb{L}^D$ be vectors respectively coding the decorations of the prefix of height $h-1$ of $T[w_0\cdots w_{i},h]$ and $T[w_0\cdots w_{i+1},h]$. Assume that $T[w_0\cdots w_{i},h]=T_j$ and $T[w_0\cdots w_{i+1},h]=T_{j'}$ for some $j,j'\in\{1,\ldots,r(h)\}$. Then there is a matrix $M_{jj'}\in\mathbb{L}^{D\times D}$ such that $M_{jj'}z_i=z_{i+1}$.
\end{lemma}
\begin{proof}
Observe that $T[w_0\cdots w_{i},h]$ contains $T[w_0\cdots w_{i+1},h-1]$ as a factor.
Let $m$ be the number of words of length less than $h-1$ in $D$.

We start by building the first $m$ rows of $M_{jj'}$.
The nodes of the first $h-2$ levels of $T[w_0\cdots w_{i+1},h-1]$ are nodes of $T[w_0\cdots w_{i},h-1]$. Their decorations are thus found in $z_i$. This means that the first $m$ rows of $M_{jj'}$ are made of zero entries and at most one entry equal to $1$. 
More precisely, for a word $u$ with $|u|\le h-2$, the entry $x_{w_0\cdots w_{i+1}u}$ corresponding to $u$ in $z_{i+1}$ is equal to the one $x_{w_0\cdots w_{i} \cdot w_{i+1}u}$ corresponding to $w_{i+1}u$ in $z_i$ and $|w_{i+1}u|\le h-1$.

Now, we build the remaining rows of $M_{jj'}$. The nodes of the last level of $T[w_0\cdots w_{i+1},h-1]$ are the nodes of the last level of $T[w_0\cdots w_{i},h]$. By assumption ($T_\mathbf{x}(L)$ is $h$-linear), we know that the decorations of these nodes can be obtained by some linear relations among decorations of the nodes on the upper levels found in $z_i$. The last rows of $M_{jj'}$ encode these relations.
\end{proof}

\begin{definition}[Matrices]
For all $j,j'\in\{1,\ldots,r(h)\}$, we let $M_{jj'}\in\mathbb{L}^{D\times D}$ denote the matrix defined in Lemma~\ref{lem:matrices}.
\end{definition}


\subsection{Regular abstract numeration systems}

We first tackle the case of numeration systems having a regular language. We start off with the example of the Fibonacci numeration system, as we will see that the general case behaves similarly.

\begin{example}
In Example~\ref{exa:fib}, we have three types of trees of height $2$: $t_0,t_1,t_2$ (see Figure~\ref{fig:ex2}). Following a path from the root of $T_\mathbf{x}(F)$ labeled by a valid representation, the only possible transitions (i.e., the type of trees that can follow another tree in a sequence of trees of height $2$ given by \eqref{eq:seqtree}) is: $t_0\to t_1$, $t_1\to t_2$, $t_2\to t_2$ and $t_2\to t_1$.  The corresponding matrices are respectively
  $$M_{01}=\begin{pmatrix}
  0&0&1\\
  -1&0&2\\
  0&0&0\\
\end{pmatrix},\ M_{12}=
\begin{pmatrix}
  0&1&0\\
  2&0&0\\
  2&0&0\\
\end{pmatrix},\ M_{22}=
\begin{pmatrix}
  0&1&0\\
  -1&2&0\\
  2&0&0\\
\end{pmatrix},\ M_{21}=
\begin{pmatrix}
  0&0&1\\
  0&0&\frac{3}{2}\\
  0&0&0\\
\end{pmatrix}.
$$
Here we let $M_{ij}$ denote the matrix that permits us to get the vector coding the decorations of level less than $2$ in a tree of type $t_j$ from the one coding the decorations of a tree of type $t_i$.
For instance, if $T[u,2]=t_1$ for some word $u\in\{0,1\}^*$, then $T[u0,2]=t_2$ and we have
\[
\begin{pmatrix}
  x_{u0}\\
  x_{u00}\\
  x_{u01}\\
\end{pmatrix}
=
\begin{pmatrix}
  x_{u0}\\
  2 x_{u}\\
  2 x_{u}\\
\end{pmatrix}
=
\underbrace{\begin{pmatrix}
  0&1&0\\
  2&0&0\\
  2&0&0\\
\end{pmatrix}}_{=M_{12}}
\begin{pmatrix}
  x_{u}\\
  x_{u0}\\
  x_{u1}\\
\end{pmatrix},
\]
by using the relations in Figure~\ref{fig:ex2}.

To compute $x_n$ for all $n\ge 1$, write the Fibonacci representation of $n$ and read the most significant digit first, and then follow the multiplication of matrices applied to the initial column vector coding the prefix of height $1$ of $T_\mathbf{x}(F)$. As an example, ten is written as $10010$ in the Fibonacci numeration system and gives the product 
$$M_{21}M_{22}M_{12}M_{01} \begin{pmatrix}
  1\\
  0\\
  2\\
\end{pmatrix}=
\begin{pmatrix}
  6\\ 9\\ 0\\
\end{pmatrix}
.$$
Indeed, first the vector coding the height-$1$ prefix of $T_\mathbf{x}(F)$ is 
\[
\begin{pmatrix}
  x_\varepsilon\\
  x_0\\
  x_1\\
\end{pmatrix}
=
\begin{pmatrix}
  1\\
  0\\
  2\\
\end{pmatrix}.
\]
Now following in $T(F)$ the path labeled by $10010$, we see the trees $t_0=T[\varepsilon,2]$, $t_1=T[1,2]$, $t_2=T[10,2]$, $t_2=T[100,2]$ in this particular order, so we multiply the initial column vector first by $M_{01}$, then by $M_{12}$ and finally by $M_{22}$ because we move from $t_0$ to $t_1$, then from $t_1$ to $t_2$ and finally, from $t_2$ to itself. With these three (rightmost) multiplications, we get the vector coding the prefix of height $1$ of $T[100,2]$. In particular, we know the decorations of $x_{1000}$ and $x_{1001}$. Since we want the one of $x_{10010}$, an extra multiplication is needed. The next tree necessarily has type $t_1$, so we have a fourth multiplication by $M_{21}$ and we get the vector coding the decorations of $T[1001,1]$. So a multiplication by $
\begin{pmatrix}
  0&1&0
\end{pmatrix}$, which is the characteristic vector for the suffix $0$ of the Fibonacci representation of ten, gives the expected decoration: $x_{10010}=9$.
Alternatively, we could have applied a fifth multiplication by $M_{12}$ to get the vector coding the decorations of $T[10010,1]$ and then a multiplication by $
\begin{pmatrix}
  1&0&0
\end{pmatrix}$, which is the characteristic vector for the remaining suffix $\varepsilon$.
\end{example}

Let us stress the importance of having here a regular numeration language. Actually, when the numeration language is regular (which is true for $F$), for large enough $h$ (and thanks to Lemma~\ref{lem:h+1} we can assume that this is the case), the undecorated factors of height $h$ are in one-to-one correspondence with the states of the minimal automaton of the numeration language (see Remark~\ref{rem:isoTw}). There is no need to separately obtain the sequence of types of trees in \eqref{eq:seqtree} that is seen during a computation. This information is put inside a finite automaton isomorphic to the minimal automaton of the numeration language where labels carry the matrix to apply --- similarly to GIFS (graph iterated function systems, see for instance \cite{Boore}). The initial state is initialized with the vector coding the decorations of the prefix $T[\varepsilon,h-1]$.  Each processed digit, so each transition, gives rise to a multiplication of the current vector by a matrix on the left. We add an output function for a final multiplication by a convenient characteristic row vector.
The Fibonacci case is depicted in Figure~\ref{fig:gd1}.
\begin{figure}[h!tb]
  \centering
  \begin{tikzpicture}[>=stealth',shorten >=1pt,auto,node distance=2cm]
    
    \node[state] (q0)               {$t_0$};
    \node (q) [left of=q0] {};    
  \node (q0') [below of=q0] {};
  \node[state]         (q1) [right of=q0] {$t_1$};
  \node (q1') [below of=q1] {};
  \node[state]         (q2) [right of=q1] {$t_2$};
  \node (q2') [below of=q2] {};
  \path[->]          (q0)  edge                    node {$1,M_{01}$} (q1);
  \path[->]          (q1)  edge   [bend left=20]   node {$0,M_{12}$} (q2);
  \path[->]          (q2)  edge   [bend left=20]   node {$1,M_{21}$} (q1);
  \path[->]          (q2)  edge   [loop right] node {$0,M_{22}$} (q2);
  \path[->] (q0) edge node {$\substack{\varepsilon:\ (1\ 0\ 0)}$} (q0');
  \path[->] (q) edge node {$
    \begin{pmatrix}
      1\\0\\2\\
    \end{pmatrix}
$} (q0);
  \path[->] (q1) edge node {$\substack{\varepsilon:\ (1\ 0\ 0)\\ 0:\ (0\ 1\ 0)}$} (q1');
  \path[->] (q2) edge node {$\substack{\varepsilon:\ (1\ 0\ 0)\\ 0:\ (0\ 1\ 0)\\ 1:\ (0\ 0\ 1)}$} (q2');
\end{tikzpicture}
  \caption{Graph directed multiplications in the Fibonacci numeration system.}
  \label{fig:gd1}
\end{figure}
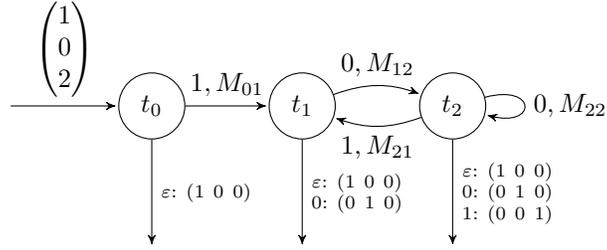

What we have seen in the example about the Fibonacci numeration system is general for any abstract numeration system built on a prefix-closed regular language. In the special case of the $k$-ary system, the automaton accepting the language~\eqref{eq:Lk} is restricted to two states.

\begin{remark}
  We can compare our construction with the one of \cite{CCS}.
The linear representation $(\lambda,\mu,\gamma)$ of the series $\sum_{w\in F} x_w w$ is given by 
\[
\lambda=
\begin{pmatrix}
 1 & 0 & 1 & 0\\
\end{pmatrix},
\mu_0=
\left(
\begin{array}{cccc}
 0 & -1 & 0 & 1 \\
 1 & 2 & 0 & 0 \\
 0 & 1 & 0 & -1 \\
 0 & 0 & 0 & 0 \\
\end{array}
\right), 
\mu_1=
\left(
\begin{array}{cccc}
 0 & 0 & 0 & 0 \\
 0 & 0 & 0 & 0 \\
 2 & 3 & 0 & 0 \\
 2 & 3 & 0 & 0 \\
\end{array}
\right),
\gamma=
\begin{pmatrix}
 1 \\ 0 \\ 0 \\ 0
\end{pmatrix}.
\]  
Observe that these matrices have dimension $4$ whereas those in the previous example have dimension $3$, but we need the extra information given by the directed graph in Figure~\ref{fig:gd1}.
  
\end{remark}

\subsection{Rational base}
We now consider the case of the rational base numeration system $(L_\frac{p}{q},A_p,<)$. Almost all the above reasonings remain valid. If $T_\mathbf{x}(L_{\frac{p}{q}})$ is $h$-linear, the set $D$ is $A_p^{<h}$. We can code the decorations of the prefix of height $h-1$ by vectors in $\mathbb{L}^D$. Since $L_\frac{p}{q}$ is non-regular, we do not have the same behavior as in the regular case. Nevertheless, each word in $A_p^h$ belongs to the domain of exactly one of the $q^h$ factors of height $h$ occurring in $L_\frac{p}{q}$. Thus the knowledge of $h$ consecutive letters $w_i\cdots w_{i+h-1}$ in a $\frac{p}{q}$-representation unambiguously determines the tree to consider and thus the matrix to apply.

\begin{definition}
We let $M_{w_i\cdots w_{i+h-1}}$ denote the matrix which, applied to the column vector coding the decorations of the prefix of height $h-1$ of $T[w_0\cdots w_{i-1},h]$, gives those of the prefix of height $h-1$ of $T[w_0\cdots w_{i},h]$.  
\end{definition}

With Example~\ref{exa:rat32}, $h=2$ and $D=\{\varepsilon<0<1<2\}$, so we have $9$ matrices of the form:
$$M_{02}=M_{00}=\left(
\begin{array}{cccc}
 0 & 1 & 0 & 0 \\
 0 & 1 & 0 & 0 \\
 0 & 0 & 0 & 0 \\
 0 & 0 & 0 & 1 \\
\end{array}
\right),\ M_{01}=\left(
  \begin{array}{cccc}
 0 & 1 & 0 & 0 \\
 0 & 0 & 0 & 0 \\
 0 & \frac{1}{2} & 0 & \frac{1}{2} \\
 0 & 0 & 0 & 0 \\
\end{array}
\right),
$$
$$M_{10}=M_{12}=\left(
\begin{array}{cccc}
 0 & 0 & 1 & 0 \\
 0 & 0 & 1 & 0 \\
 0 & 0 & 0 & 0 \\
 -2 & 0 & 3 & 0 \\
\end{array}
\right), \
M_{11}=\left(
\begin{array}{cccc}
 0 & 0 & 1 & 0 \\
 0 & 0 & 0 & 0 \\
 -1 & 0 & 2 & 0 \\
 0 & 0 & 0 & 0 \\
\end{array}
\right),$$
$$M_{20}=M_{22}=\left(
\begin{array}{cccc}
 0 & 0 & 0 & 1 \\
 0 & 0 & 0 & 1 \\
 0 & 0 & 0 & 0 \\
 0 & -1 & 0 & 2 \\
\end{array}
\right)
,\ M_{21}=\left(
\begin{array}{cccc}
 0 & 0 & 0 & 1 \\
 0 & 0 & 0 & 0 \\
 -\frac{1}{2} & 0 & 0 & \frac{3}{2} \\
 0 & 0 & 0 & 0 \\
\end{array}
\right).$$
For instance, suppose that $u\in A_3^*$ can be followed by $02$ in a base-$\frac{3}{2}$ representation.
Then $M_{02}$ codes the relation between the vectors
\[
\begin{pmatrix}
  x_{u}\\x_{u0}\\x_{u1}\\x_{u2}\\
\end{pmatrix}
\text{ and }
\begin{pmatrix}
  x_{u0}\\x_{u00}\\x_{u01}\\x_{u02}\\
\end{pmatrix}
\]
when $T[u,2]$ is the third tree on the first row of Figure~\ref{fig:ex3} (the tree is determined by the length-$2$ word indexing $M_{02}$).
Looking at this particular tree allows us to fill in the entries of $M_{02}$:
\[
\begin{pmatrix}
  x_{u0}\\x_{u00}\\x_{u01}\\x_{u02}\\
\end{pmatrix}
=
 \begin{pmatrix}
  x_{u0}\\x_{u0}\\0\\x_{u2}\\
\end{pmatrix}
=
\underbrace{\left(\begin{array}{cccc}
 0 & 1 & 0 & 0 \\
 0 & 1 & 0 & 0 \\
 0 & 0 & 0 & 0 \\
 0 & 0 & 0 & 1 \\
\end{array}
\right)}_{=M_{02}}
\begin{pmatrix}
  x_{u}\\x_{u0}\\x_{u1}\\x_{u2}\\
\end{pmatrix}.
\]

The prefix of height $1$ of $T_\mathbf{s}(L_{\frac{3}{2}})$ is coded by the vector $
\begin{pmatrix}
  0&0&0&2\\
\end{pmatrix}^T$ because $x_\varepsilon=0$, $x_2=2$ and $0,1$ are not valid $\frac{3}{2}$-representations.

Now take the representation $2120012$ of the integer $22$ in base $\frac{3}{2}$. We compute
$$ M_{12}M_{01}M_{00}M_{20}M_{12}M_{21}\begin{pmatrix}
  0\\0\\0\\2\\
\end{pmatrix}=
\begin{pmatrix}
  6\\6\\0\\8\\
\end{pmatrix},
$$
as sliding a length-$2$ window in the word $2120012$ gives the six matrix indices $21, 12, 20, 00, 01, 12$.
We get the decorations of the prefix of height $1$ of the tree $T[212001,2]$. Since the remaining suffix is $2$, we apply a last multiplication by a characteristic row vector $
\begin{pmatrix}
  0&0&0&1\\
\end{pmatrix}$ to get $x_{2120012}=8$.

\begin{remark}
  In the above example, we could even use vectors and matrices of dimension $3$ using the fact that a node $u$ cannot simultaneously have children $u0,u2$ and $u1$. Our matrices and vectors always have a zero row. 
\end{remark}

As a conclusion, we have shown that one way to extend the notion of regular sequences is to consider linear decorated numeration trees. We have not considered the formalism of rational series. Let $k\ge 2$ be an integer. 
As a reminder \cite{Ring,BR2}, a sequence $f:\mathbb{N}\to\mathbb{K}$ is $k$-regular if the formal series $\sum_{w\in A_k^*} f(\val_k(w))\, w$ is $\mathbb{K}$-recognizable (for a convenient choice of a semiring $\mathbb{K}$). One can equivalently consider a series whose support is $A_k^*\setminus 0A_k^*$. In a similar way, one could ask if there exists a function $f:\mathbb{N}\to\mathbb{K}$ such that the series $\sum_{w\in A_p^*} f(\val_{\frac{p}{q}}(w))\, w$ is $\mathbb{K}$-recognizable and whose support is $L_\frac{p}{q}$.

\section*{Acknowledgment}
Manon Stipulanti is supported by the FNRS Research grant 1.B.397.20.


\begin{thebibliography}{99}

\bibitem{Akiyama--Frougny-Sakarovitch-2008}
S. Akiyama, Ch. Frougny, and J. Sakarovitch, Powers of rationals modulo 1 and rational base number systems, {\em Israel J. Math.} {\bf 168} (2008), 53--91.

\bibitem{Ring} J.-P. Allouche and J. Shallit, The ring of $k$-regular sequences, {\em Theoret. Comput. Sci.} {\bf 98} (1992), 163--197.

\bibitem{AS2}  J.-P. Allouche and J. Shallit, The ring of $k$-regular sequences II, {\em Theoret. Comput. Sci.} {\bf 307} (2003), 3--29.

\bibitem{Boore} G. C. Boore and K. J. Falconer, Attractors of directed graph IFSs that are not standard IFS attractors and their Hausdorff measure, {\em Math. Proc. Cambridge Philos. Soc.} {\bf 54} (2013), 325--349. 

\bibitem{BR} J. Berstel and C. Reutenauer, Recognizable formal power series on trees, {\em Theoret. Comput. Sci.} {\bf 18} (1982), 115--148. 

\bibitem{BR2} J. Berstel and C. Reutenauer, {\em Noncommutative Rational Series with Applications},  Encyclopedia of Math. and its Appl. {\bf 137}, Cambridge Univ. Press (2010).

  
\bibitem{CRS} \'E. Charlier, N. Rampersad and J. Shallit, Enumeration and decidable  properties  of  automatic  sequences, {\em Int. J. Found.  Comput. Sci.} \textbf{23} (2012), 1035--1066.
  
\bibitem{CCS} C. Cisternino, \'E. Charlier, and M. Stipulanti, Regular sequences and synchronized sequences in abstract numeration systems, arXiv:2012.04969.

\bibitem{LR} P. Lecomte and M. Rigo, Numeration systems on a regular language, {\em Theory Comput. Syst.} {\bf 34} (2001), 27--44. 

\bibitem{LRS} J. Leroy, M. Rigo, and M. Stipulanti, Counting the number of non-zero coefficients in rows of generalized Pascal triangles, {\em Discrete Math.} {\bf 340} (2017), 862--881.

\bibitem{Marsault-these} V. Marsault, {\em \'Enum\'eration et num\'eration}, PhD thesis, T\'elecom-Paristech, 2015.

\bibitem{Marsault--Sakarovitch-2} V. Marsault and J. Sakarovitch,
  Breadth-first serialisation of treesand rational languages,
  Developments in Language Theory - 18th International Conference,
  2014, Ekaterinburg, Russia, August 26-29, 2014, {\em Lect. Notes
    Comp. Sci.} {\bf 8633}, 252--259.

\bibitem{Marsault--Sakarovitch-2016} V. Marsault and J. Sakarovitch, Trees and languages with periodic signature, {\em Indagationes Mathematicae} {\bf 28} (2017),  221--246.

\bibitem{MR} M. Rigo and A. Maes, More on generalized automatic sequences, {\em 
J. Autom. Lang. Comb.} {\bf 7} (2002), 351--376. 

\bibitem{RS} M. Rigo and M. Stipulanti, Automatic sequences: from rational bases to trees, {\tt arxiv.2102.10828} preprint (2021).

\bibitem{Rowland} E. Rowland, What is an automatic sequence?, \emph{Notices of the AMS} {\bf 62} (2015), 274--276.

\bibitem{IS} E. Rowland, IntegerSequences: a package for computing with $k$-regular sequence, {\tt https://ericrowland.github.io/packages.html}

\bibitem{RowlandSti} E. Rowland and M. Stipulanti, Avoiding $5/4$-powers on the alphabet of nonnegative integers, \textit{Electron. J. Combin.} \textbf{27} (2020), Paper 3.42, 39 pp.

\bibitem{Sha} J. Shallit,  Remarks  on  inferring  integer  sequences, \url{https://cs.uwaterloo.ca/~shallit/Talks/infer.ps}
  
\bibitem{OEIS} N. Sloane et al., The On-Line Encyclopedia of Integer Sequences,
\url{http://oeis.org}.  
\end{thebibliography}
\end{document}